\documentclass[twoside,leqno]{article}

\usepackage[letterpaper]{geometry}

\usepackage{ltexpprt}

 \usepackage{algorithm}
 \usepackage[noend]{algpseudocode}

 \algnewcommand\algorithmicswitch{\textbf{switch}}
 \algnewcommand\algorithmiccase{\textbf{case}}

 \algdef{SE}[SWITCH]{Switch}{EndSwitch}[1]{\algorithmicswitch\ #1\ \algorithmicdo}{\algorithmicend\ \algorithmicswitch}%
 \algdef{SE}[CASE]{Case}{EndCase}[1]{\algorithmiccase\ #1}{\algorithmicend\ \algorithmiccase}%
 \algtext*{EndSwitch}%
 \algtext*{EndCase}%

\usepackage[linktocpage]{hyperref}
\hypersetup{
    unicode=false,          
    colorlinks=true,        
    linkcolor=red,          
    citecolor=OliveGreen,        
    filecolor=magenta,      
    urlcolor=cyan           
}

\usepackage[utf8]{inputenc}
\usepackage[dvipsnames]{xcolor}

\usepackage[utf8]{inputenc}
\usepackage{amsfonts,amssymb}
\usepackage{mathtools}
\usepackage{cite}
\usepackage{url}
\usepackage{appendix}
\usepackage{bbm}
\usepackage{thm-restate}
\usepackage[capitalize, nameinlink]{cleveref}
\usepackage{enumitem}

\crefname{@theorem}{Theorem}{Theorems}
\crefname{@theorem}{Theorem}{Theorems}
\crefname{Remark}{Remark}{Remarks}

\newcommand{\eps}{\varepsilon}

\newcommand{\congest}{$\mathsf{CONGEST}$\xspace}

\newcommand{\poly}{\operatorname{\text{{\rm poly}}}}

\newcommand{\diam}{\operatorname{diam}}

\newcommand{\sep}{\operatorname{sep}}
\newcommand{\load}{\operatorname{load}}

\newcommand{\dist}{\operatorname{dist}}

\newcommand{\hop}{\operatorname{hop}}

\newcommand{\ball}{\operatorname{ball}}

\global\long\def\poly{\mathrm{poly}}%
\global\long\def\eps{\epsilon}%
\global\long\def\tO{\tilde{O}}%
\global\long\def\cN{{\cal N}}%
\global\long\def\cS{{\cal S}}%
\global\long\def\congest{\mathrm{cong}}%
\global\long\def\sep{\mathrm{sep}}%
\global\long\def\val{\mathrm{val}}%
\global\long\def\dist{\mathrm{dist}}%
\global\long\def\deg{\mathrm{deg}}%
\global\long\def\diam{\mathrm{diam}}%
\global\long\def\cov{\mathrm{cov}}%
\global\long\def\ball{\mathrm{ball}}%
\renewcommand{\poly}{\operatorname{poly}}

\usepackage{makecell}

\usepackage{bbm}
\newcommand{\Ind}[1]{\mathbbm{1}(#1)}

\newcommand{\cB}{\mathcal{B}}
\newcommand{\cG}{\mathcal{G}}
\newcommand{\cW}{\mathcal{W}}

\renewcommand{\vec}[1]{\boldsymbol{#1}}
\newcommand{\mat}[1]{\boldsymbol{#1}}

\newcommand{\vp}{\vec{p}}
\newcommand{\vq}{\vec{q}}
\newcommand{\vr}{\vec{r}}
\newcommand{\vw}{\vec{w}}
\newcommand{\mP}{\mat{P}}

\global\long\def\tOmega{\tilde{\Omega}}%

\allowdisplaybreaks

\newcommand{\norm}[1]{||#1||}

\begin{document}

\title{\Large A Cut-Matching Game for Constant-Hop Expanders\thanks{The author ordering was randomized using \url{https://www.aeaweb.org/journals/policies/random-author-order/generator} (Code: evspzRM8Zg-n). It is requested that citations of this work list the authors separated by \texttt{\textbackslash textcircled\{r\}} instead of commas: Haeupler \textcircled{r} Huebotter \textcircled{r} Ghaffari.}}
\author{Bernhard Haeupler\thanks{INSAIT, Sofia University ``Kliment Ohridski'' \& ETH Zurich; \texttt{bernhard.haeupler@inf.ethz.ch}; Partially funded by the Ministry of Education and Science of Bulgaria's support for INSAIT as part of the Bulgarian National Roadmap for Research Infrastructure and through the European Research Council (ERC) under the European Union's Horizon 2020 research and innovation program (ERC grant agreement 949272).}
\and Jonas Huebotter\thanks{ETH Zurich} \and Mohsen Ghaffari\thanks{MIT}}

\date{}

\maketitle




\fancyfoot[R]{\scriptsize{Copyright \textcopyright\ 2025\\
Copyright for this paper is retained by authors}}








\begin{abstract} 
This paper extends and generalizes the well-known cut-matching game framework and provides a novel cut-strategy that produces constant-hop expanders.

 \medskip

Constant-hop expanders are a significant strengthening of regular expanders with the additional guarantee that any demand can be (obliviously) routed along constant-hop flow-paths --- in contrast to the $\Omega(\log n)$-hop paths in expanders. 

\medskip

Cut-matching games for expanders are key tools for obtaining linear-time approximation algorithms for many hard problems, including finding (balanced or approximately-largest) sparse cuts, certifying the expansion of a graph by embedding an (explicit) expander, as well as computing expander decompositions, hierarchical cut decompositions, oblivious routings, multi-cuts, and multi-commodity flows. 

\medskip

The cut-matching game of this paper is crucial in extending this versatile and powerful machinery to constant-hop and length-constrained expanders~\cite{LCexpanderAlg24} and has been already been extensively used\footnote{
A distance-matching game by Chuzhoy~\cite{CSODA2023} is closely related to this work. Both were developed independently in 2022 and made public at the same time. Their implications is still not fully understood but as far as the (very imperfect) understanding of all authors is concerned the distance-matching game has direct related powerful applications on its own but cannot be used in the efficient algorithms for computing length-constrained expanders decompositions~\cite{LCexpanderAlg24} for general graphs with capacities and length and their algorithmic applications cited above.}. For example, as a key ingredient in several recent breakthroughs, including, computing constant-approximate $k$-commodity (min-cost) flows in $(m+k)^{1+\epsilon}$ time~\cite{haeupler2024lowstep} as well as the optimal constant-approximate deterministic worst-case fully-dynamic APSP-distance oracle~\cite{LCdynamic24} --- in all applications the constant-approximation factor directly traces to and crucially relies on the expanders from a cut-matching game guaranteeing constant-hop routing paths.
\end{abstract}

\newpage

{  \small
   \bigskip
   \hypersetup{linkcolor=blue}
   \tableofcontents
}



\newpage






\section{Introduction}\label{sec:introduction}

This paper reformulates, extends, and generalizes the well-known cut-matching game framework~\cite{KRV,KKOV} and provides a novel cut-strategy that produces constant-hop expanders~\cite{haeupler2022hop,LCexpanderAlg24}. Constant-hop expanders are constant-diameter regular expanders with the additional guarantee that routing paths for every multi-commodity demand can be chosen to only consist of constantly many edges -- instead of the polylogarithmic length guaranteed for normal expanders. 

We first give an introduction to cut-matching games, the extensive prior work on it, and their many applications. Any expert reader can skip this and go straight to Section~\ref{sec:constant-hop-expanders}, which briefly describing what constant-hop expanders are and why they are important, or Section~\ref{sec:results}, which explains our results. 

\subsection{Cut-Matching Games}

Cut-matching games were first introduced through the work of Khandekar, Rao and Vazirani~\cite{KRV} as a key tool for achieving highly-efficient approximation algorithms for the intensely studied sparsest-cut problem. We present here the formulation given by Khandekar, Khot, Orecchia, and Vishnoi~\cite{KKOV}:

A cut-matching game is played by two players and starts with an empty $n$-node multi-graph $G$. In each round,%
\begin{itemize}
    \item The cut player chooses a bisection $(S, \bar{S})$ of the vertices.
    \item The matching player then chooses a perfect matching $M$ between $S$ and $\bar{S}$.
    \item The matching $M$ is added to $G$. If $G$ is now a sufficiently strong expander the cut player wins and the game terminates. Otherwise, the game continues with the next round.
\end{itemize}

Research has been primarily concerned with strategies for the cut-player that guarantee strong expansion properties for $G$ after a small number of rounds. Such guarantees have to hold no matter how the adversarial matching player plays, i.e., no matter what perfect matchings across the cuts (which are adaptively generated by the cut-strategy) are added to $G$. A cut-strategy is efficient if the next bisection can be efficiently computed, given the matchings of previous rounds. Ideally the cut-strategy is deterministic and has a close-to-linear running time $T_{\text{cut-player}}(n)$.

These desiderata stem from the fact~\cite{KRV} that any efficient cut-player that guarantees $G$ to be a $\phi$-expander\footnote{A graph $G$ is a $\phi$-expander iff every cut in $G$ is at most $\phi$-sparse, i.e., the number of edges leaving a set $S$ is at least $\phi$ times the volume  of $S$ (or $\bar{S}$).} after $r$ rounds can be directly transformed into an  $O(\frac{1}{\phi})$-approximation algorithm for the sparsest-cut problem. The running time of this sparsest-cut algorithm is $O(r \cdot (T_{\text{max-flow}}(n,m) + T_{\text{cut-player}}(n)))$ where $T_{\text{max-flow}}(n,m)$ is the running time of computing an ($O(1)$-approximate) single-commodity flow. Given that approximate max-flows in undirected graphs can be computed in nearly linear time~\cite{christiano2011electrical,sherman2013nearly,kelner2014almost,peng2016approximate,chen2022maximum},  this running time is dominated by $r \cdot T_{\text{cut-player}}$, i.e., the number of rounds and the efficiency of the cut-player.
This reduction is conceptually very simple, and we give an informative high-level explanation of it in Section \Cref{sec:KRVreduction}. We recommend \Cref{sec:KRVreduction} to any reader not already familiar with how cut-matching games are used to reduce sparsest-cut and expander embedding problems to single-commodity flows. We remark that the new cut-matching game is applied in a similar manner, too. Full technical details and complete proofs for applications and how to make the cut-matching strategy presented in this paper constructive are given in \cite{LCexpanderAlg24}.

\medskip

\noindent{\bfseries Prior Results on Cut-Matching Games and Approximating the Sparsest-Cut via Flow Computations}

Khandekar, Rao and Vazirani~\cite{KRV} gave an efficient cut-strategy (KRV) that guaranteed an $\Omega(\frac{1}{\log^2 n})$-expander after $r=O(\log^ 2 n)$ rounds. This KRV-strategy is randomized, works with high probability, and can be computed in linear time. Using their reduction (explained in \Cref{sec:KRVreduction}) they also turned their cut-strategy into a randomized $O(\log^2 n)$-approximation algorithm for the sparsest-cut problem using merely $r = O(\log^2 n)$ flow computations. This result~\cite{KRV} started a long line of works reducing the problem of approximating the sparsest-cut to computing few single-commodity flows. Prior approaches~\cite{AroraHazanKale,AroraKale} aimed at lowering the running time of the $O(n^{9.5})$-time state-of-the-art SDP-based $O(\sqrt{\log n})$-approximation algorithm of Arora, Rao and Vazirani~\cite{arora2009expander} required computing much harder multi-commodity flows.

Improvements by Arora and Kale~\cite{AroraKale} and Orecchia, Schulman, Vazirani and Vishnoi~\cite{OSVV} showed how to compute an $O(\log n)$-approximate sparsest cut using a polylogarithmic number of single-commodity flows by building on the ideas of \cite{KRV}. Khandekar, Khot, Orecchia, Vishnoi~\cite{KKOV} explicitly studied the above cut-matching game framework and proved the existence of a cut-strategy producing a $O(\frac{1}{\log n})$-expander in $O(\log n)$ rounds. 
This KKOV cut-strategy does not have an efficient implementation because it requires the cut-player to find the min-bisection or at least an (approximately-)sparsest (approximately-)balanced cut. By noting that this is exactly one of the problems cut-matching games can be used for, and applying a recursive approach to resolve this cyclicity, Chuzhoy
Gao, Li, Nanongkai, Peng, and Saranurak~\cite{chuzhoy2020deterministic} gave an efficient cut-strategy based on KKOV with weaker guarantees. The advantage of this KKOV-based cut-strategy is that it is deterministic, whereas the (directly computationally efficient) approach of KRV seems inherently randomized. Indeed, the  $m^{o(1)}$-approximation for balanced sparse-cuts of \cite{chuzhoy2020deterministic} and its generalization to capacitated graphs by Li and Saranurak~\cite{li2021weightedExpanderDecomposition} remain the state-of-the-art deterministic algorithms for this problem. The cut-strategy presented in this paper is deterministic, can be seen as a generalization of KKOV, and uses the recursion idea of \cite{chuzhoy2020deterministic} to make the cut-player efficient.

Lower bounds for the power of the cut-matching framework have also been proven. Particularly,
\cite{KKOV} proved that cut-matching games could not give a better than $\Omega(\sqrt{\log \log n})$-approximation and \cite{OSVV} improved this to $\Omega(\sqrt{\log n})$. Sherman~\cite{sherman2009breaking} proved an $\Omega(\frac{\log n}{\log \log n})$ lower bound on the approximation ratio achievable via the cut-matching framework, assuming flow is not re-routed. This proves that the KKOV cut-strategy is essentially best-possible. Sherman~\cite{sherman2009breaking} also gave a $O(\sqrt{\log n})$-approximation algorithm for the sparsest-cut problem which uses $O(n^{\eps})$ single-commodity flow computations.

\medskip

\noindent{\bfseries Applications and Algorithms for which Cut-Matching Games are Crucial.}

In addition to providing close-to-linear time approximation algorithms for the sparsest-cut problem, cut-matching games have proven to be versatile, flexible, and powerful tools that can be used for many other (multi-commodity) flow and cut problems. For many of these problems, cut-matching games are the only known tool that leads to close-to-linear-time algorithms. We list a selection here.

\begin{itemize}
    \item Finding a balanced or more generally approximately-largest $\phi$-sparse cut.
    \item Computing expander decompositions: Given a graph $G$ and $\phi$ an expander decomposition deletes at most a $\phi$-fraction of edges in $G$ such that all connected components of the resulting graph are $\tO(\phi)$-expanders. Expander decompositions have been crucial tools in a vast number of recent results, advances, and breakthroughs. Notably, they have been used for fast Laplacian solvers~\cite{spielman2004nearly,cohen2017almost}, almost-linear time maximum flows~\cite{kelner2014almost,chen2022maximum}, bipartite matchings~\cite{van2020bipartite}, graph sparsification~\cite{spielman2011spectral,chu2020graph,chalermsook2021vertex}, and dynamic algorithms~\cite{nanongkai2017dynamic_a,wulff2017fully,nanongkai2017dynamic_b,chuzhoy2020deterministic,bernstein2020fully,bernstein2020deterministic,goranci2021expander}.

    \item Certifying the expansion $\phi$ of a graph $G$ by computing an embedding of a $\tOmega(1)$-expander into $G$ with congestion $\tO(\frac{1}{\phi})$. (see also \Cref{sec:KRVreduction})

    \item Expander routing: Given a demand between pairs of nodes in an expander, find (short) routes between these pairs that minimize the congestion. While there exist several randomized ways to obtain such an expander routing~\cite{ghaffari2018new, ghaffari2017distributed}, the only deterministic algorithm producing such routings is the direct reduction to (efficient and deterministic) cut-matching games of Chang and Saranurak~\cite{chang2020deterministic}. Expander routings can also be used to embed any chosen explicit expander (e.g., a hypercube) into any arbitrary expanding graph. This strengthens the certification stated above.

    \item Hierarchical cut decompositions: These are useful tools for achieving approximation algorithms for a wide variety of network design problems~\cite{racke2002minimizing, bienkowski2003practical, harrelson2003polynomial}. R{\"a}cke, Shah, and T{\"a}ubig~\cite{racke2014computing} show how to compute such decompositions in near-linear time using cut-matching games. 

    \item Oblivious routings: These are distributions over paths between any two nodes with the strong property that any demand can be routed with $O(\log n)$-competitive congestion by simply sampling each routing path independently and therefore obliviously with respect to the competing traffic/demand. Initial polynomial-time constructions for such oblivious routings used LP-duality and low-stretch spanning trees~\cite{racke2008optimal}. It remains unknown whether this approach can be made to run in sub-quadratic time. All close-to-linear-time algorithms for oblivious routings use either hierarchical cut decompositions~\cite{racke2014computing} or hierarchical expander decompositions~\cite{goranci2021expander}, both of which crucially rely on cut-matching games.

    \item Recent deterministic fully-dynamic algorithms for connectivity and other fundamental problems~\cite{chuzhoy2020deterministic, goranci2021expander} rely on expander hierarchies and the low-recourse dynamic algorithms maintaining the expander hierarchies crucially rely on the close-to-linear-time balanced sparse cut routines provided by cut-matching games.

    \item The state-of-the-art near-linear time algorithm for computing $(1+\eps)$-approximate maximum flows in undirected graphs~\cite{peng2016approximate} crucially relies on the RST hierarchical cut decomposition~\cite{racke2014computing} and cut-matching games. Also, the recent breakthrough giving an almost-linear time maximum flow algorithm for directed graphs~\cite{chen2022maximum} uses cut-matching games.

\end{itemize}

To summarize, for many of the above problems, either the only close-to-linear time algorithms crucially use cut-matching games, or cut-matching games provide the state-of-the-art solution.

\subsection{Constant-Hop Expanders}\label{sec:constant-hop-expanders}

The constant-hop expanders for which this paper provides a cut-matching game and cut-strategy, were first introduced by Haeupler \textcircled{r} Raecke \textcircled{r} Ghaffari~\cite{haeupler2022hop}. More generally, \cite{haeupler2022hop} introduced $(h,s)$-hop $\phi$-expanders. As we discuss in \Cref{sec:hop_expanders}, such graphs are essentially characterized by being sufficiently well-connected to route any demand between close-by nodes with small congestion over short paths, i.e., any $h$-hop unit-demand can be routed over $O(hs)$-hop paths with congestion $O(\frac{\log n}{\phi})$.

Of course this is exactly what $\phi$-expanders do, and indeed a $(h,s)$-hop $\phi$-expander is simply a $\phi$-expander for any $h = \Omega(\frac{\log n}{\phi})$, because it is well-known that the congestion-competitive routes in an expander can be chosen to be of length $\Theta(\frac{\log n}{\phi})$. Hop-constrained expanders become interesting for smaller $h$ because then they extend the powerful techniques and machinery developed for expanders to suddenly control not just $\ell_\infty$-quantities like cuts, flows, and congestion but simultaneously and separately also $\ell_1$-quantities like path-lengths and costs.

\cite{haeupler2022hop} developed a comprehensive theory and foundation for such expanders and demonstrated that much of the algorithmic expander machinery, including, efficient expander decompositions, hierarchical applications of such expander decompositions to give hop-constrained equivalents of hierarchical cut decompositions and expander hierarchies, can be transferred. For example, \cite{haeupler2022hop} gives close-to-linear time constructions of hop-constrained oblivious routings~\cite{ghaffari2021hop}. Further work strengthens and expands this connection~\cite{LCexpanderAlg24, haeuplerparallelgreedyspanners}.

During this effort, another unique and powerful property of hop-constrained expanders has become apparent, namely, that separating losses in length and congestion make it possible to talk about sub-logarithmic path-length and $\ell_1$-losses. Indeed, even expanders with a conductance almost arbitrarily close to the best possible value of $1/2$ still require $\Theta(\log n)$-long paths. In contrast, hop-constrained expanders with constant $(h,s)$
and expander decompositions\footnote{Naturally, $t = \frac{1}{\eps}$ routing paths require node degrees of $n^{O(\eps)}$ and indeed often lead to (provably necessary) $n^{O(\eps)}$ losses in congestion. While such losses might seem quite large, they are actually very tolerable and interesting in most settings. Indeed, one often can completely eliminate such losses in approximation factors by using standard boosting methods based on multiplicative-weights to shift these $n^{O(\eps)}$ losses into the running time~\cite{LCexpanderAlg24}.}
for them can be seen to exist when looking more closely at \cite{haeupler2022hop}. We call such $(t,O(1))$-hop expanders with constant $t$ simply \emph{constant-hop expanders}.

While \cite{haeupler2022hop} shows the existence of constant-hop expanders and expander decompositions, the complete algorithmic machinery of hop-constrained expanders uses the KRV cut-matching games for regular expanders to achieve almost-linear time algorithms. This leads to large (super-)\!~polylogarithmic losses in the length of flow paths. Generally, it is clear that inherently no known cut-matching game can be used for constant-hop expanders because such games at their very best lead to $\Theta(\log n)$ path lengths over which commodities are mixed and routed.

\subsection{Our Results}\label{sec:results}

This paper crucially expands the algorithmic theory for length-constrained expanders from \cite{haeupler2022hop} and points out one of their very useful parameterizations: there exist \emph{constant-hop expander decompositions}, which guarantee that any (local) demand in $G-C$ can be competitively routed over constant-hop paths.

Constant-hop expander decompositions and cut-matching games are particularly important and interesting because they circumvent poly-logarithmic approximation barriers inherent in expanders and expander-based algorithms. For example, hierarchical cut decompositions~\cite{harrelson2003polynomial,racke2014computing,goranci2021expander} achieve a $O(\log^2 n)$\,-\!~approximation to the cut-structure of a graph and decompositions with somewhat higher approximations can be computed in near-linear time~\cite{racke2014computing,goranci2021expander,peng2016approximate}. While it is an open question whether the $O(\log^2 n)$-approximation can be reduced, it is a priori clear that it cannot be pushed below $\Theta(\log n)$. This also limits the approximation that can be achieved by close-to-linear-time algorithms for multi-cuts and $\Theta(n)$-commodity flows when these algorithms are based on expanders.

Unfortunately, cut-matching games for expanders are the only known way to efficiently construct many expander-based structures and obtain close-to-linear-time algorithms, including (hop-constrained) expander decompositions. Such cut-matching games however cannot be used for constant-hop expanders. The result of this paper, a cut-matching game for constant-hop expanders, fills this gap:

\smallskip

{\bfseries Cut-Matching Game for Constant-Hop Expanders (informal, see \Cref{cor:cut_strategy}):}\\
For any sufficiently large (constant) $t = O(\frac{\log n}{\log \log n})$, there exists a $k = n^{O(\frac{1}{t})}$ and a cut-strategy that runs in $t$ iterations, produces $k$ cuts, and results in a constant-hop expander with degree at most $k$ and diameter at most $t$. This expander can route any unit-demand along $t$-hop paths with congestion at most $k$.

\smallskip

We extend and reformulate the cut-matching game framework to allow for multi-cuts (with polynomially many parts) and constant number of parallel rounds each adding polynomially many edges to most nodes while mixing is guaranteed to only ever happen over constant-hop paths throughout all rounds. 

Our generalized framework and novel cut-strategy in this new framework is clean and simple (in hindsight!) and with the right perspective can be phrased as a natural generalization of the KKOV cut-matching game~\cite{KKOV}. Before this work it was completely unclear if a cut-matching framework for constant-hop expanders was possible and how it could look like. A major conceptual contribution is the simple set-up described in this paper that allows the new cut-strategy be phrased as close-as-possible to the well-known cut-matching games for regular expanders.

Fortunately, while the results in this paper are somewhat involved and technical, these complexities can be almost entirely pushed to the analysis. Like KKOV, we use the entropy-potential to prove progress and correctness. However, while the KKOV analysis is essentially trivial once it is set up correctly, the definition of the random walks and mixing processes in our analysis are much less obvious. Furthermore, proving that these processes increase the entropy by at least $\Theta(\frac{1}{t} \cdot n \log n)$ in each of the $t$ iterations constitutes the main technical contribution of this paper.

\medskip

\noindent{\bfseries Efficiency and Extensions}
Our cut-strategy can be implemented computationally efficient, albeit with slightly weaker parameters. Full proofs and details are given in \cite{LCexpanderAlg24}, which also gives a computationally efficient expander decomposition algorithms for length-constrained expanders using the cut-matching game from this paper as a crucial ingredient. The astute reader might realize that the algorithms in \cite{LCexpanderAlg24} themselves rely on the cut-matching game presented in this paper, which seems cyclic at first. 

However, as explained above, the same is true for KKOV. Indeed, the KKOV cut-strategy requires the cut-player to find and play a roughly balanced sparse-cut. Note that a regular expander decomposition can be used for this. Either such an expander decomposition contains a large component certifying that no balanced sparse cut exists, or any balanced grouping of connected components induces a sparse cut. Again, the (deterministic) close-to-linear-time algorithms for computing either an expander decomposition or a balanced sparse-cut typically rely on the KKOV cut-strategy. This was resolved in~\cite{chuzhoy2020deterministic} through a simple and elegant recursive approach.

The same recursion idea from \cite{chuzhoy2020deterministic} works here to make the new cut-strategy efficient (see \cite{LCexpanderAlg24} for details): For this, one simply guarantees that, instead of the constant-hop expander decomposition on an $m$-sized graph calling $\tO(1)$ cut-matching games with $\tO(m)$ nodes, it calls at most $\tO(k)$ cut-matching games with at most $m/k$ nodes. This reduced size guarantees that this recursion completes quickly, for example within constant depth $O(1/\eps)$ for $k=n^{\eps}$ where $\eps$ is a constant.


The cut-strategy presented here has a number of nice properties and further advantages: For one, just like KKOV, it is deterministic and lends itself to deterministic algorithms. Secondly, the well-separated clusterings and expander decompositions used by the cut-player also directly work for capacitated graphs with edge lengths and arbitrary node weightings. This is important in many settings and particularely for the design of efficient approximation algorithms for such general weighted graphs. Lastly, as we show in the technical part, the progress analysis of our strategy continues to work without change if the matching-player provides large matchings, instead of perfect matchings. This means that our cut-strategy naturally identifies sparse-cuts that are balanced, if they exist. More generally, it naturally identifies approximately-largest sparse cuts.

Overall, hop-constrained expanders show great promise to drastically expand the powerful expander-based techniques that have been used to great effect over the last decade. The constant-hop expanders first highlighted in this work, in particular, enable these techniques to talk about routes that are of constant-length and constant-competitive in length compared to optimal routes, something that seems completely out of reach for techniques based on regular expanders. As we briefly outline next subsequent works and breakthrough results building on constant-hop expanders have indeed proven this promise to hold up. 

\medskip

\noindent{\bfseries Parallel and Independent Work~\cite{CSODA2023}}

A parallel and independent work by Chuzhoy~\cite{CSODA2023} gives a ``distance matching game'' for well-connected graphs. 

The concept of $(\nu,d)$-well-connected graphs, proposed subsequently to \cite{haeupler2022hop} but developed independently, is similar in spirit to $(h,s)$-hop (subset) $\phi$-expanders when restricted to graphs with diameter less than $h$ (such that every demand is $h$-hop). There are some technical differences between the two structures, but both provide separate control over congestion ($\nu$ and $\frac{1}{\phi}$) and length ($d$ and $hs$) of routing paths. Furthermore, like the focus on small or constant $s$ in this paper \cite{CSODA2023} is particularly interested in routing via very short sub-logarithmic paths.

One important difference between well-connected graphs and $(h,s)$-hop expanders  seems to be that $(h,s)$-hop expanders are designed to be an expander-like notion that applies to arbitrary graphs with large diameter in which every (degree-restricted) demand between $h$-close nodes can be routed with ($\tilde{O}(\frac{1}{\phi})$-)low  congestion via ($hs$-)short flow paths, where $h$ and $hs$ are both independent and typically much smaller than the diameter of the $(h,s)$-hop expander. As such well-connected graphs seem to be essentially equivalent and correspond more closely to what subsequent works like \cite{haeuplerparallelgreedyspanners,haeupler2024lowstep} call routers. Overall, it seems unclear if or how the notion of an $(h,s)$-hop expander decomposition of a general graph that stands at the center of \cite{haeupler2022hop} can be related to the tools and notions \cite{CSODA2023} develops for well-connected graphs. 

On the other hand, \cite{CSODA2023} provides a wide-variety of direct applications including a novel deterministic close-to-linear-time implementation of the KKOV cut-matching game player for regular expanders that achieves a $O(\log^{2+O(1/\eps)} n)$-approximation instead of the previously best efficient $\log^{O(1/\eps)} n$ approximation.
Other applications include an $m^{(1 + \eps + o(1))}$ algorithms for APSP in expanders. 

Subsequent to \cite{CSODA2023} and this work Chuzhoy and Zhang~\cite{chuzhoy2023new} furthermore gave a $(\log \log n)^{2^{O(1/\eps^3)}}$-approximate fully-dynamic deterministic all-pair-shortest-path distance oracle with amortized $n^{\eps}$ update time based on well-connected graphs, the RecDynNC reduction of \cite{chuzhoy2021decremental}, and the distance-matching game~\cite{CSODA2023}. 

\medskip

\noindent{\bfseries Subsequent Results and Applications of This Work}

The cut-matching game for constant-hop expanders presented in this paper and its entropy potential analysis has subsequently been one of the key components in making constant-hop expanders algorithmically tractable and has enabled several important results. 

Subsequent works building on and crucially relying on constant-hop expanders include the constant-factor approximation for multi-cut, and (min-cost) $k$-commodity flows of \cite{haeupler2024lowstep} with a $(m+k)^{1+O(\eps)}$ running time, which is close-to-linear even for $k=m$ commodities and thus shatters the $\Theta(mk)$ multi-commodity flow decomposition barrier. 

The cut-matching game for constant-hop expanders developed in this paper is also one of many  ingredients fueling the breakthrough result of~\cite{LCdynamic24} which provides a fully-dynamic deterministic all-pair-shortest-path distance oracles with $n^{\eps}$ update time and achieving an essentially and conditionally-optimal~\cite{abboud2022hardness} constant $O_\eps(1)$-factor approximation. This distance oracle furthermore features a worst-case (and parallel) update time bypassing for the first time the ubiquitous but inherently sequential and amortized  Even-Shiloach tree data structure~\cite{EvenS}.

As a last example, constant-hop expanders have also been used to provide new notions for constant-stretch fault-tolerant spanners~\cite{bodwin2024fault}. 

In all three applications the achieved constant-approximations are crucial to the strength of the results and trace directly back to the constant-hopness guaranteed by the cut-matching game developed in this paper.

\section{Preliminaries}

In this paper, a graph $G$ is generally an undirected uncapacitated unweighted multi-graph with unit edge-lengths.
We use $V_G$ to denote the vertex set of $G$ and $n_G = |V_G|$ for the number of vertices.
We use $E_G \subseteq \binom{V_G}{2}$ for the edge set and $m_G =|E_G|$ for the number of edges.
We use $\Gamma_G(v) = \{u \in V : \{u, v\} \in E_G\}$ to denote the neighborhood of a vertex $v \in V$.
The degree of a vertex $v$ is $\deg_{G}(v)= |\Gamma_G(v)|$.
When the correspondence to $G$ is clear from context, we drop the subscript.

Given subsets of vertices $S, T \subseteq V_G$, we denote by $E_G(S, T)$ the set of edges with one endpoint in $S$ and one endpoint in $T$, and call this the \emph{cut between $S$ and $T$}.

In \cref{sec:hop_expanders}, we collect results from prior work regarding hop-expanders and hop-constrained expander decompositions which will be used in our analysis.
In \cref{sec:well_separated_clustering}, we introduce well-separated clusterings which we will use to find ``well-separated'' subsets of vertices, each with small diameter.

\subsection{Hop-Constrained Expanders and Expander Decompositions}\label{sec:hop_expanders}

\mbox{ }
\smallskip

\noindent{\bfseries Paths and Length}

A path from vertex $v$ to vertex $w$ is called a $(v,w)$-path.
For any path $P$ we define the length of $P$ to be the number of edges in $P$ and denote it with $|P|$. The \emph{distance} between vertices $v$ and $w$ is $\dist_{G}(v,w)=\min_{P:(v,w)\text{-path}}|P|$.

A \emph{ball} of radius $r$ around a vertex $v$ is $\ball_{G}(v,r)=\{w\mid\dist_{G}(v,w)\le r\}$.
The \emph{diameter} of a set of vertices $U \subseteq V$ is $\diam_{G}(U)=\max_{v,w\in U} \dist_{G}(u,v)$.

\medskip

\noindent{\bfseries (Length-Constrained) Multi-Commodity Flows}

A \emph{(multi-commodity) flow} $F$ in $G$ is a function that assigns simple paths $P$ in $G$ a flow value $F(P)\ge0$. We say $P$ is a \emph{flow-path} of $F$ if $F(P)>0$. $P$ is a \emph{$(v,w)$-flow-path} of $F$ if $P$ is both a $(v,w)$-path and a flow-path of $F$. The
$(v,w)$-flow $f_{(v,w)}$ of $F$ is such that $f_{(v,w)}(P)=F(P)$ if $P$ is a $(v,w)$-path , otherwise $f_{(v,w)}(P)=0$. The \emph{value}
of $f_{(v,w)}$ is $\val(f_{(v,w)})=\sum_{P\text{ is a }(v,w)\text{-path}}F(P)$. The value of $F$ is $\val(F) = \sum_{P} F(P)$.

The \emph{congestion of $F$ on an edge $e$} is $\congest_{F}(e)=F(e)$ (as we assume graphs to be uncapacitated) where $F(e)=\sum_{P:e\in P}F(P)$ denotes the total flow value of all paths going through $e$.
The \emph{congestion} of $F$ is $\congest_{F}=\max_{e\in E(G)}\congest_{F}(e)$.
The \emph{length} of $F$, denoted by $\hop_{F}=\max_{P:F(P)>0}|P|$,
measures the maximum length of all flow-paths of $F$.

\medskip

\noindent{\bfseries Demands}

All demands in this paper are directed, non-negative, and real-valued demands between nodes.

A \emph{demand} $D:V\times V\rightarrow\mathbb{R}_{\ge0}$ assigns a non-negative value $D(v,w) \ge 0$ to each ordered pair of vertices.
The size of a demand is written as $|D|$ and defined as $\sum_{v,w} D(v,w)$.
A demand $D$ is called \emph{$h$-hop} if it assigns positive values only to pairs that are within distance at most $h$, i.e., $D(v,w)>0$ implies that $\dist_{G}(v,w)\le h$.

Given a flow $F$, the \emph{demand routed/satisfied by $F$} is denoted by $D_{F}$ where, for each $u,v\in V$, $D_{F}(u,v)=\sum_{P\text{ is a }(u,v)\text{-path}}F(P)$ is the value of $(u,v)$-flow of $F$.
We say that a \emph{demand $D$ is routeable in $G$ along $t$-hop paths with congestion $\eta$} if there exists a flow $F$ in $G$ where $D_{F}=D$, $\congest_{F}\le\eta$, and $\hop_{F}\le t$.

We say that a demand $D$ is a \emph{unit-demand} if $\max\{\sum_{w \in V}D(v,w),\sum_{w \in V}D(w,v)\} \leq 1$ for all $v \in V$. Additionally, we say that $D$ is a \emph{unit-demand on $S \subset V$} if $D$ is a unit-demand and $\max\{\sum_{w \in V}D(v,w),\sum_{w \in V}D(w,v)\} = 0$ for all $v \not\in S$.
We remark that in contrast to prior works \cite{haeupler2022hop}, we define unit-demands with respect to vertices rather than with respect to edges. Within bounds considered negligible in the context of this paper, all graphs have small and roughly uniform degrees.

\medskip

\noindent{\bfseries Hop-Constrained Expanders}

The general notion of $(h,s)$-hop $\phi$-expanders where $h$ is called \emph{hop-length}, $s$ is called \emph{length slack}, and $\phi$ is called \emph{conductance} is defined in terms of so-called ``moving cuts'' and can be found in \cite{haeupler2022hop}. The exact details of their definition will not be important for us here.

All we will need is that up to a constant length factor and a poly-logarithmic congestion factor, $(h,s)$-hop $\phi$-expanders for unit-demands are equivalent to the notion of routing any $h$-hop unit-demand along $(s \cdot h)$-hop paths with congestion roughly $\frac{1}{\phi}$. This is made precise by the following theorem:

\begin{theorem}
[Flow Characterization of Hop-Constrained Expanders (Lemma 3.16 of \cite{haeupler2022hop})]\label{thm:flow character}For any $h \geq 1$, $\phi < 1$, and $s \geq 1$ we have that:
\begin{enumerate}
\item If $G$ is an $(h,s)$-hop $\phi$-expander for unit-demands, then every $h$-hop unit-demand can be routed in $G$ along $(s \cdot h)$-hop paths with congestion $O(\log(n)/\phi)$.
\item If $G$ is not an $(h,s)$-hop $\phi$-expander for unit-demands, then some $h$-hop unit-demand cannot be routed in $G$ along $(\frac{s}{2}\cdot h)$-hop paths with congestion $1/2\phi$.
\end{enumerate}
\end{theorem}

\medskip

\noindent{\bfseries Hop-Constrained Expander Decompositions}

A \emph{(pure) cut} $C \subseteq E_G$ is a subset of edges and is said to have size $|C|$.
We use $G - C = (V_G, E_G \setminus C)$ to denote the \emph{graph $G$ after cutting $C$}.

\begin{Definition}
[Hop-Constrained Expander Decomposition]An \emph{$(h,s)$-hop $(\phi,\kappa)$-expander
decomposition} with respect to unit-demands in a graph $G$ is a (pure)
cut $C$ of size at most $h s \kappa \phi n$ such that $G - C$ is an $(h,s)$-hop
$\phi$-expander for unit-demands.
\end{Definition}

We call $s$ and $\kappa$ the length slack and congestion slack, respectively. The existence of good $(h,s)$-hop $\phi$-expander decompositions is proven in \cite{haeupler2022hop}. We state a slightly different parametrization below obtained by simply scaling the exponential demand in the proof of \cite{haeupler2022hop}. While \cite{haeupler2022hop} also proves stronger expander decompositions for sub-set expanders and for “moving cuts” which allow a congestion slack independent from $h$ we only state the simplified version for pure cuts here, as it is easier to state and sufficient for our purposes.

\begin{theorem}[Existence of Hop-Constrained Expander Decompositions~\cite{haeupler2022hop}]
\label{thm:expdecomp_exist} For any length bound
$h$, length slack $s\ge100$, conductance bound $\phi>0$, there exists an $(h,s)$-hop $(\phi,\kappa)$-expander
 decomposition with respect to unit-demands in $G$ with congestion
 slack $\kappa\le h \cdot n^{O(1/s)}\log n$.
 \end{theorem}


\subsection{Well-separated Clusterings}\label{sec:well_separated_clustering}

Well-separated clusterings provide a method of grouping a graph into ``well-separated'' clusters such that each cluster has small diameter.

Given a graph $G$ with lengths (or equivalently any metric on a vertex set $V$), a \emph{clustering} $\cS$ in $G$ with diameter $h_{\diam}$ is a collection of mutually disjoint vertex sets $S_1,\ldots,S_{|\cS|}$, called clusters, where every cluster has (weak) diameter at most $h_{\diam}$, i.e., $\forall i: \ \diam_{G}(S_i) \le h_{\diam}$. A clustering furthermore has separation $h_{\sep}$ if the distance between any two clusters is at least $h_{\sep}$.

A well-separated clustering $\cN$ with width $w$ is a collection of $w$ many clusterings $\cS_1, \ldots, \cS_w$ such that for every node $v$ there exists a cluster $S_{i,j}$ such that $v \in S_{i,j}$.

A well-separated clustering $\cN$ has diameter $h_{\diam}$ if every clustering has diameter at most $h_{\diam}$ and separation $h_{\sep}$ if this applies to each of its clusterings.

Lastly, we denote with $\load_{\cN}$ the maximum number of clusters (or equivalently, clusterings) any node $v$ is contained in.

We remark that well-separated clusterings are simply neighborhood covers with covering radius $1$.
The following two results are a slight adaptation of results for neighborhood covers from \cite{miller2013parallel}.

\begin{theorem}
\label{thm:cover-absolute-separation-existential}
For any $h_{\cov}, h_{\sep}, \eps, \eps'$ and graph $G$ there exists a well-separated clustering with separation $h_{\sep}$, diameter $h_{\diam} = \frac{h_{\sep}}{\eps \eps'} h_{\cov}$, width $w = \frac{n^{O(\eps)}}{\eps} \log n$, $\load_{\cN} = n^{O(\eps')} \log n$.
\end{theorem}

\begin{theorem}\label{thm:cover-absolute-separation-alg}
There exists a deterministic CONGEST algorithm
that given $\eps, \eps'$ and a separation $h_{\sep}$ computes a well-separated clustering with the parameters provided in \Cref{thm:cover-absolute-separation-existential}, in $h_{\diam} n^{O(\eps + \eps')}$ rounds and total work $\tO(n) \cdot \frac{n^{O(\eps + \eps')}}{\eps} \cdot h_{\diam}$. There also exists a deterministic PRAM algorithm with depths $\tO(n^{O(\eps + \eps')})$
and total work $\tO(n) \cdot n^{O(\eps + \eps')}$
and $h_{\diam} \cdot n^{O(\eps + \eps')}$ rounds.
\end{theorem}

\section{Entropy of Pseudo-Distributions: Splitting, Merging, and Mixing}

In this section, we develop some notation and results for pseudo-distributions and their entropy. We mainly prove results on how certain splitting, merging, and mixing operations affect the entropy of pseudo-distributions. While the new cut-matching algorithm itself does not perform any such operations, the results here are needed for its potential analysis. 

We use the natural logarithm throughout this paper and denote it with $\log$.
We call a vector $\vq \in [0,1]^S$ a \emph{pseudo-distribution} on the finite set $S$.
We often assume, without loss of generality, that the support set $S$ is consecutively numbered, i.e., $S = [k] = \{1,\ldots,k\}$ with $k=|S|$.
Given a subset $S' \subseteq S$, we write $\vq(S') = \sum_{i \in S'} \vq(i)$.
We call $\norm{\vq} = \vq(S)$ the size of $\vq$.
A pseudo-distribution $\vq$ is a \emph{distribution} if it has size $1$.
We denote the \emph{entropy} of a pseudo-distribution $\vq$ with $H(\vq) = \sum_{i \in S} H(\vq(i))$ where $H(u) = - u \log u$.



\begin{lemma}\label{lem:entropy_splitting}
    Given pseudo-distributions $\vp, \vq \in [0,1]^k$, we have that $H(\vp + \vq) \leq H(\vp) + H(\vq)$.
\end{lemma}
\begin{proof}
    We have \begin{align*}
        H(\vp + \vq) &= - \sum_{i=1}^k \vp(i) \log(\vp(i) + \vq(i)) + \vq(i) \log(\vp(i) + \vq(i)) \\
        &\leq - \sum_{i=1}^k \vp(i) \log\vp(i) + \vq(i) \log\vq(i) = H(\vp) + H(\vq)
    \end{align*} where the inequality uses $\vp(i) \leq \vp(i) + \vq(i)$ and $\vq(i) \leq \vp(i) + \vq(i)$.
\end{proof}

\subsection{Splitting and Merging}\label{sec:splitting_and_merging}

\begin{lemma}[Splitting lower bound]\label{lem:splitting}
    Let $\vp \in [0,1]^k$ be a pseudo-distribution and let $\vq$ be a pseudo-distribution obtained by splitting each entry $\vp(i)$ into entries $\vq(i, j)$ where $\sum_j \vq(i, j) = \vp(i)$. If every split entry is at most a $\gamma$-fraction for some $\gamma \geq 1$, then the entropy increases by at least $\norm{\vp}\log\gamma$. That is, if ${\vq(i, j) \leq \frac{\vp(i)}{\gamma}}$ for every $i,j$, then $H(\vq) \geq H(\vp) + \norm{\vp}\log\gamma$.
\end{lemma}
\begin{proof}
    We have \begin{align*}
        H(\vq) = - \sum_{i=1}^k \sum_j \vq(i, j) \log \vq(i, j) &\geq \sum_{i=1}^k \log\left(\frac{\gamma}{\vp(i)}\right) \underbrace{\sum_j \vq(i, j)}_{= \vp(i)} \\
        &= H(\vp) + \log\gamma \cdot \underbrace{\sum_{i=1}^k \vp(i)}_{=\norm{\vp}} = H(\vp) + \norm{\vp}\log\gamma
    \end{align*} where the inequality uses $\vq(i, j) \leq \frac{\vp(i)}{\gamma}$.
\end{proof}

\begin{lemma}[Splitting upper bound]\label{lem:splitting_ub}
    Let $\vp \in [0,1]^k$ be a pseudo-distribution and let $\vq$ be a pseudo-distribution obtained by splitting each entry $\vp(i)$ into at most $\gamma$ entries $\vq(i, j)$ where $\sum_j \vq(i, j) = \vp(i)$. Then, $H(\vq) \leq H(\vp) + \norm{\vp}\log\gamma$.
\end{lemma}

\begin{proof}

    Suppose, we split each entry $i$ of $\vp$ into exactly $\gamma$ entries with uniform size. Then, \begin{align*}
        H(\vq) = \sum_{i=1}^k \sum_{j=1}^\gamma \frac{\vp(i)}{\gamma} \log \frac{\gamma}{\vp(i)} =
        \sum_{i=1}^l \vp(i) \left(\log \frac{1}{\vp(i)} + \log \gamma\right) =        H(\vp) + \norm{\vp} \log \gamma.
    \end{align*}
This split maximizes the entropy of $\vq$ because the entropy of $\vq$ is the sum of the entropies of all splits and for every split of $\vp(i)$ the uniform split into $\gamma$ many entries of size $\frac{\vp(i)}{\gamma}$ maximizes the entropy among all pseudo-distributions with size $\vp(i)$ and support of size $\gamma$. For any other split we therefore have $H(\vq) \leq H(\vp) + \norm{\vp} \log \gamma$.
\end{proof}

We will largely use the ``converse''  of \cref{lem:splitting_ub} to lower bound the entropy loss of a pseudo-distribution $\vp$ obtained by ``merging'' entries of some pseudo-distribution $\vq$ in the following manner:

For a pseudo-distribution $\vq$ on $S$ and a partition $S_1, \ldots, S_{k'} \subset S$ of $S$ we define the merged pseudo-distribution $\vp$ by setting $\vp(i)=\sum_{j\in S_i} \vq(j)$ for every $i$. Note that $\norm{\vp}=\norm{\vq}$. We call $\gamma = \max_{i \in [k']} |S_i|$ the merging factor, i.e., every entry of $\vp$ is the result of merging the probability mass of at most $\gamma$ entries in $\vq$.


\begin{lemma}[Merging lower bound]\label{lem:merging}
    Suppose that a pseudo-distribution $\vp$ is the merging of a pseudo-distribution $\vq$ with merging factor at most $\gamma$. Then, 
    $H(\vp) \geq H(\vq) - \norm{\vq}\log\gamma$.
\end{lemma}
\begin{proof}
Since $\vp$ can be seen as splitting each entry of $\vq$ into at most $\gamma$ entries this is merely a restatement of \Cref{lem:splitting_ub}.
\end{proof}

\subsection{Mixing Processes with Stable Entropies}\label{sec:mixing_process_stable_entropy}

In the following, we study the entropy change of certain random walks. We say that a random walk has \emph{stable entropy} if the entropy of any state distribution can only increase or stay the same after taking one step of the random walk. We show that certain random walks we use in our analysis, which we call ``two-step mixing processes'', have stable entropy.

Let $\vp$ be a pseudo-distribution on a weighted set $V$ with weights $\vw \in \mathbb{N}^{|V|}$.

Consider any collection $\cW \subseteq \mathcal{P}(V)$ of so-called ``mixers'' $W \subseteq V$ such that for every $v \in V$, $|\cW(v)| = \vw(v)$ where $\cW(v)$ is the set of mixers $W \in \cW$ such that $v \in W$, i.e., any $v$ is part of exactly $\vw(v)$ mixers.
We write $\gamma_v = 1 / \vw(v)$.

We consider mixing processes with the following dynamics: \begin{enumerate}
    \item The probability mass at each $v \in V$ is ``split'' uniformly between all mixers including $v$.
    Formally, consider a pseudo-distribution $\vq$ on the set of mixers $\cW$ wherein $\vq(W) = \sum_{v \in W} \gamma_v \vp(v)$. Note that $\norm{\vq} = \norm{\vp}$.

    \item Then, each mixer $W$ returns ``mixed'' probability mass proportional to the weight of each $v \in W$. Formally, we write $\gamma_W = \sum_{v \in W} \gamma_v$ and have \begin{align*}
        \vp'(v) = \sum_{W \in \cW(v)} \frac{\gamma_v}{\gamma_W} \vq(W).
    \end{align*}
\end{enumerate}




We call such a mixing process a \emph{two-step mixing process}. 

\begin{lemma}\label{lem:preserving_size}
    $\vp$ and $\vp'$ are of the same size.
\end{lemma}
\begin{proof}
    Observe that the total mass at mixer $W$ after the mixing \begin{align*}
        \sum_{v \in W} \frac{\gamma_v}{\gamma_W} \vq(W) =  \vq(W) \sum_{v \in W} \frac{\gamma_v}{\gamma_W} =  \vq(W)
    \end{align*} is identical to the mass at mixer $W$ before the mixing. Hence, \begin{align*}
        \norm{\vp'} = \sum_{v \in V} \vp'(v) = \sum_{v \in V} \sum_{W \in \cW(v)} \frac{\gamma_v}{\gamma_W} \vq(W) &= \sum_{W \in \cW} \sum_{v \in W} \frac{\gamma_v}{\gamma_W} \vq(W) = \sum_{W \in \cW} \vq(W) = \norm{\vq} = \norm{\vp}.
    \end{align*}
\end{proof}

We now show that any such two-step mixing process has stable entropy.

\begin{lemma}\label{lem:entropy_conv_comb}
    For any $x_1, \dots, x_k > 0$ and $\theta_1, \dots, \theta_k \geq 0$ such that $\theta_1 + \dots + \theta_k = 1$, \begin{align*}
        H\left(\sum_{i=1}^k \theta_i x_i\right) \geq \sum_{i=1}^k \theta_i H(x_i).
    \end{align*}
\end{lemma}
\begin{proof}
    We have \begin{align*}
        -\left(\sum_{i=1}^k \theta_i x_i\right) \log \left(\sum_{i=1}^k \theta_i x_i\right) + \sum_{i=1}^k \theta_i \left(x_i \log x_i\right) &= -\sum_{i=1}^k \theta_i x_i \log\left(\frac{\sum_{j=1}^k \theta_j x_j}{x_i}\right) \\
        &= -\sum_{i=1}^k \theta_i x_i \log\left(\theta_i + \sum_{j \neq i} \frac{\theta_j x_j}{x_i}\right) \\
        &\geq \sum_{i=1}^k \theta_i x_i - \theta_i^2 x_i - \sum_{j \neq i} \theta_i \theta_j x_j \\
        &= \sum_{i=1}^k \theta_i x_i \underbrace{(1 - \theta_i - \sum_{j \neq i} \theta_j)}_{=0} = 0
    \end{align*} where the inequality uses $\theta_i + \sum_{j \neq i} \frac{\theta_j x_j}{x_i} > 0$, $\log(1 + x) \leq x$ for any $x > -1$, and the final equality reorganizes the terms of the sum.
\end{proof}

Consider the pseudo-distribution $\tilde{\vq}$ on the set of mixers $\cW$ defined by $\tilde{\vq}(W) = \vq(W) / \gamma_W$.

\begin{lemma}\label{lem:stable_entropy:1}
    For any $W \in \cW$, $H(\tilde{\vq}(W)) \geq \frac{1}{\gamma_W} \sum_{v \in W} \gamma_v H(\vp(v))$.
\end{lemma}
\begin{proof}
    Observe that $\tilde{\vq}(W) = \sum_{v \in W} \frac{\gamma_v}{\gamma_W} \vp(v)$ is a convex combination. Hence, by \cref{lem:entropy_conv_comb}, \begin{align*}
        H(\tilde{\vq}(W)) \geq \sum_{v \in W}\frac{\gamma_v}{\gamma_W} H(\vp(v)).
    \end{align*}
\end{proof}

\begin{lemma}\label{lem:stable_entropy:2}
    For any $v' \in V$, $H(\vp'(v')) \geq \sum_{W \in \cW(v')} \frac{\gamma_{v'}}{\gamma_W} \sum_{v \in W} \gamma_v H(\vp(v))$.
\end{lemma}
\begin{proof}
    Observe that $\vp'(v') = \sum_{W \in \cW(v')} \gamma_{v'} \tilde{\vq}(W)$ is a convex combination. Hence, \begin{align*}
        H(\vp'(v')) \geq \sum_{W \in \cW(v')} \gamma_{v'} H(\tilde{\vq}(W)) \geq \sum_{W \in \cW(v')} \frac{\gamma_{v'}}{\gamma_W} \sum_{v \in W} \gamma_v H(\vp(v))
    \end{align*} where the first inequality follows from \cref{lem:entropy_conv_comb} and the last inequality follows from \cref{lem:stable_entropy:1}.
\end{proof}

\begin{theorem}\label{thm:stable_entropy_mixing_process}
    Any two-step mixing process has stable entropy.
\end{theorem}
\begin{proof}
    We have \begin{align*}
        H(\vp') = \sum_{v' \in V} H(\vp'(v')) &\geq \sum_{v' \in V} \sum_{W \in \cW(v')} \frac{\gamma_{v'}}{\gamma_W} \sum_{v \in W} \gamma_v H(\vp(v)) \\
        &= \sum_{W \in \cW} \sum_{v \in W} \gamma_v H(\vp(v)) = \sum_{v \in V} H(\vp(v)) = H(\vp)
    \end{align*} where the inequality is due to \cref{lem:stable_entropy:2} and the second equality is due to an analogous argument to the proof of \cref{lem:preserving_size}.
\end{proof}







\section{The Cut-Matching Game for Constant-Hop Expanders}

We consider the following general framework of cut-matching games.

A cut-matching-game on a set of vertices $V$ with $r$ rounds
produces a sequence of graphs $G_0,\dots,G_r$ on the node set $V$ where $G_0$ is the empty graph.
There are two players, the cut player and the matching player.

In round $i$, the cut player chooses two disjoint subsets $S_i, T_i \subset V$ of equal size based on $G_i$. The matching player then produces a unit-capacity unit-length (perfect) matching graph $G'_i=(V,M_i)$ with edge set $M_i \subset S_i \times T_i$ such that $\deg_{G'_i}(v) = 1$ for all $v \in V$. The graph $G_{i+1}$ is then simply $G_{i} + G'_i$, i.e., $G_{i+1}=(V,E_i)$ where $E_i = \bigcup_{j \leq i} M_j$.

A \emph{cut-strategy} is a function which given a graph $G$ produces two equally-sized disjoint subsets of vertices $S_i, T_i$. Similarly, a \emph{matching-strategy} is a function which given a graph $G$ and two equally-sized subsets of vertices $S_i,T_i$ produces a valid matching graph $G'_i$.

We generally desire a cut-strategy such that if the above game is played for $r$ rounds against any matching strategy this results in a graph $G_r$ which can route any unit-demand along $t$-hop paths with congestion $\eta$ and for which $\deg_{G_r} \preccurlyeq \Delta$. Thereby, we seek $t, \eta$, and $\Delta$ to be as small as possible.


Note that in contrast to the standard cut-matching game where in each round the cut player chooses a bisection of vertices (cf. \cref{sec:introduction}), this framework allows for the cut player to play arbitrary cuts. This is motivated by the observation that routing commodity between arbitrary pairs of vertices over a constant number of hops requires $\poly(n)$ degree.\footnote{We discuss this in greater detail in the ``warm-up''.} Therefore, to construct such an expander efficiently within a constant number of iterations of the cut-matching game, one has to obtain $\poly(n)$ matchings in parallel. We refer to each such batch of cuts (combining multiple rounds) as an \emph{iteration} of the cut-matching game.

Before discussing the general cut-strategy in \cref{sec:general_cut_strategy}, we explore a simple example which will introduce the main conceptual ideas.

\subsection{Warm-up}

Consider the problem of finding a graph on a vertex set $V$ which can route any unit-demand along $t$-hop paths with congestion $\eta$.
Further, consider the idealized case where in the beginning of iteration $i$ of the cut-matching game, for the graph $G_i$ one of the following holds: either \begin{enumerate}
    \item there is a subset of vertices $S$ of size $\frac{n}{2}$ such that $G_i$ routes any unit-demand on $S$ along $(t-2)$-hop paths with congestion $\frac{\eta}{2}$; or
    \item for any $k \in [n]$, there are disjoint subsets of vertices $S_1, \dots, S_k$, each of size $\frac{n}{k}$, such that the $G_i[S_j]$ are ``well-separated'' (we return to this condition later).
\end{enumerate}

In the former case, the cut player can ensure that $G_{i+1}$ can route any unit-demand on $V$ along $t$-hop paths with congestion $\eta$ by playing the cut $(S, V \setminus S)$.
This is because then the demand $D(u,v)$ (for $u, v \not\in S$) can be routed along the path \begin{align*}
    u \xrightarrow[M_{\text{finish}}]{} u' \xrightarrow[{G_i}]{} \cdots \xrightarrow[{G_i}]{} v' \xrightarrow[M_{\text{finish}}]{} v \qquad\text{for some $u', v' \in S$}
\end{align*} where $M_{\text{finish}}$ is the matching of $(S, V \setminus S)$.
The congestion of newly added edges between $S$ and $V \setminus S$ is $1$, and the congestion of edges in $G_i$ grows at most by a factor of $2$ as the degrees grow at most by a factor $2$.

In the latter case, our cut player chooses all pairwise cuts $\{(S_j, S_{j'})\}_{j,j' \in [k], j \neq j'}$, yielding a graph $M_i$ which is a union of $k \choose 2$ perfect matchings.
We will show that this latter case can only occur for at most $t$ iterations of the cut-matching game (when $k$ is chosen appropriately) by studying the entropy-increase of a random walk on the returned matching graphs.

Before making this more concrete, observe that for a lazy random walk on an $n$-vertex graph $G$ to ``mix'' in $t$ steps, the degree $\Delta$ of vertices in $G$ must be at least $n^{\frac{1}{t}}$.
To see this, consider a particle located at any vertex $v$.
In $t$ steps of the random walk, this particle can reach at most $\Delta^t$ vertices.
As we need $\Delta^t \geq n$ for mixing, we must have $\Delta \geq n^{\frac{1}{t}}$.

We study the lazy random walk of a commodity $\nu \in [n]$ on $V$ which, in each iteration, remains at the current vertex with probability $\frac{1}{k}$ and moves to one of its $k-1$ neighbors in $M_i$ with probability $\frac{1}{k}$ each.
Let us now consider a single iteration $i$.
We denote by $\vp_{\nu}^i \in [0,1]^n$ the distribution over locations of commodity $\nu$ before iteration $i$.
We encode the joint distribution of commodities (one initially located at each vertex) by $\mP^{i} = [\vp_1^i, \dots, \vp_n^i]$.
Observe that $\mP^i$ is doubly stochastic.
The $u$-th row (denoted $\vq_u^i$) of $\mP^i$ can be seen as a probability distribution over which commodity is located at vertex $u$ before iteration $i$.
As each commodity is moved independently and the entropy $H(\vp_{\nu}^i)$ is maximized for the uniform distribution, $H(\mP^i) \leq n \log n$.
Moreover, for this idealized example, we assume for every $u \in S_j, v \in S_{j'}, j \neq j'$ that the support of $\vq_u^i$ and the support of $\vq_v^i$ are disjoint (as the subsets are ``well-separated'').
Fixing commodity $\nu$, \begin{align*}
    H(\vp_{\nu}^{i+1}) = \sum_{u = 1}^n \vp_{\nu}^i(u) \log \frac{k}{\vp_{\nu}^i(u)} = H(\vp_{\nu}^i) + \log k,
\end{align*} and therefore, $H(\mP^{i+1}) = H(\mP^i) + n \log k$.
Hence, the latter case occurs at most $(\log n) / (\log k)$ many times.
If we choose $k = n^{\frac{1}{t}}$, the cut-matching game terminates after at most $t$ iterations and the degree of vertices in the final graph is at most $\Delta \leq t k + 1 = t n^{\frac{1}{t}} + 1$.
We remark that this gives a cut-strategy with at most $t + 1$ iterations which presents at most $t {k \choose 2} + 1$ cuts to the matching player.

In the general setting, we cannot assume that there always exists a single idealized ``well-separated'' clustering that contains all nodes. Instead, we will use a set of clusterings (a well-separated clustering), within each of which clusters are ``well-separated'' and such that ``most'' nodes appear in at least one clustering (cf. \cref{sec:well_separated_clustering}). Moreover, we want that no nodes appear in too many clusterings as this would lead to a large degree in the final graph. As the clusters within a single clustering will in general not be of the same size, we need to ``round'' clusters by omitting vertices (cf. \cref{sec:clustering_decomp}).

Furthermore, we want to ensure that any cluster $S$ of the well-separated clustering has ``small'' diameter and can route any unit-demand on $S$ along $(t-2)$-hop paths with low-congestion so that if we find a ``large'' cluster, we can return a graph that routes along $t$-hop paths with low congestion within ``few'' additional rounds. To ensure this, we compute a ``$(t-2)$-hop'' expander decomposition before finding a well-separated clustering with small diameter, effectively disregarding some edges which are used to route only few demands.

We discuss the cut-strategy for the general setting in the following section.

\subsection{The General Cut-Strategy}\label{sec:general_cut_strategy}

Given a set of vertices $V$ and a sufficiently small parameter $\epsilon \geq \frac{\log \log n}{\log n}$ trading hop-length and congestion of the cut-strategy, the general cut-strategy is described in \cref{alg:cut_strategy}.
We assume that there exists an (efficient) algorithm for finding hop-constrained expander decompositions.

The cut-strategy has two phases. We call the phase corresponding to the repeat-until-loop the ``main phase'' of the algorithm, and the phase corresponding to the remaining part the ``final phase''.

\begin{algorithm}
\caption{Cut-Strategy}\label{alg:cut_strategy}
\begin{algorithmic}
    \Require vertex set $V$, sufficiently small $\epsilon \geq \frac{\log \log n}{\log n}$
    \State
    \Comment{main phase}
    \State $G \gets (V, \emptyset)$
    \Repeat
        \State $C \gets$ $(h,s)$-hop $\phi$-expander decomposition of $G$ with respect to unit-demands
        \State \ \ \ \ \ \ \ \ and with congestion slack $\kappa$
        \State $\cN = \{S_{j,j'}\}_{j \in [w], j'} \gets$ well-sep. clustering of $G - C$ with diameter $h_{\diam}$ and separation $h_{\sep}$
        \If{$\forall j,j':\ |S_{j,j'}| < \frac{n}{k'}$}
        \State $\cG = \{B_{j, i'}\}_{j \in [g], i' \in [k]} \gets$ greedily ``split \& merge'' $\cN$ into $g$ groups using \cref{alg:clustering_decomp}
        \State \ \ \ \ \ \ (each group consists of exactly $k$ equal-sized $h_{\sep}$-separated blocks)
        \State $M \gets \mathrm{MatchingPlayer}(\bigcup_{j=1}^{g} \bigcup_{i'_1, i'_2 \in [k], i'_1 \neq i'_2} \{(B_{j, i'_1}, B_{j, i'_2})\})$
        \State $G \gets G \cup M$
        \EndIf
    \Until{$\exists j,j':\ |S_{j,j'}| \geq \frac{n}{k'}$}
    \State
    \Comment{final phase}
    \State $\{T_{i}\}_{i \in [k'']} \gets$ partition of $V \setminus S_{j,j'}$ such that $|T_{i}| = |S_{j,j'}|$ for all $i < k''$ and $|T_{k''}| \leq |S_{j,j'}|$
    \State $M_{\text{finish}} \gets \mathrm{MatchingPlayer}(\{\mathrm{select}(S_{j,j'}, T_{i})\}_{i=1}^{k''})$
        \State $G_{\text{final}} \gets G \cup M_{\text{finish}}$
\end{algorithmic}
\end{algorithm}

We define $\mathrm{select}(S, T)$ where $|S| \geq |T|$ as $(S', T)$ where $S'$ is any $|T|$-element subset of $S$.\footnote{We may have for the final set $T_{k''}$ that $|T_{k''}| < |S_{j,j'}|$ in which case we simply want to match $T_{k''}$ to any subset of vertices of $S_{j,j'}$.} We denote by $\mathrm{MatchingPlayer}(\cdot)$ a function which given a set of cuts, returns a union of perfect matchings for each cut.

During the main phase, we first find an $(h,s)$-hop $\phi$-expander decomposition $G - C$ of $G$ with respect to unit-demands and with congestion slack $\kappa$. We choose $\phi$ sufficiently small such that $|C|$ is small,\footnote{This is made precise in \cref{lem:ed_leakage}.} and we choose $s \cdot h \leq t - 2$ to ensure that $G - C$ routes $h$-hop (i.e., ``local'') unit demands along $(t-2)$-hop paths with low-congestion. In the final phase, we then ``complete'' $G - C$ to a graph routing ``global'' unit-demands  along $t$-hop paths with low-congestion if we find a large-enough cluster with small diameter.

To this end, we then find a well-separated clustering $\cN$ of $G - C$ with diameter $h_{\diam} \leq h$ and with separation $h_{\sep} \geq 2 b$ where $b$ is the number of iterations during the main phase to ensure that clusters are ``well-separated'' which we will use to show that the main phase completes within ``few'' iterations.

As soon as there exists a large cluster $S_{j,j'}$ in $\cN$ within which we can route along $(t-2)$-hop paths with low congestion, we proceed to the final phase where we match all remaining vertices to this large cluster. As the cluster $S_{j,j'}$ was large to begin with, this preserves low-congestion and low-degree as only few edges are added while allowing demands between any $u,v \in V$ to be routed over $t$ steps.

If there does not exist a large cluster yet, we decompose the well-separated clustering $\cN$ into $g$ groups, each consisting of $k$ blocks such that all blocks have the same size and blocks within a single group still have pairwise separation $h_{\sep}$. We ensure that only a small fraction of vertices is not within any block, i.e., ``dropped'' from the well-separated clustering. We do so by first removing all clusterings with size much smaller than $\frac{n}{w}$ (the average size of a clustering is at least $\frac{n}{w}$), and then splitting all clusterings of a much larger size into clusterings of size roughly $\frac{n}{w}$. During both steps, we drop only a small constant fraction of vertices. Finally, we merge the clusters of each clustering into $k$ blocks and drop vertices from blocks such that blocks are of the same size.
Fundamentally, this is possible while only dropping a small fraction of vertices because the size of the largest cluster $\frac{n}{k'}$ is chosen to be tiny compared to the size of ``most'' clusterings.
The results of this clustering decomposition are summarized in \cref{thm:clustering_decomp}. Proofs can be found in \cref{sec:clustering_decomp}.

\begin{theorem}[Clustering Decomposition]\label{thm:clustering_decomp}
    Given a well-separated clustering $\cN = \{S_{j,j'}\}_{j \in [w], j'}$ on vertices $V$ with width $w$ and separation $h_{\sep}$ such that $|S_{j,j'}| \leq \frac{n}{k'}$, and any $c,c' \in (0,1), k \in \mathbb{N}$ such that $k' \geq \frac{w k}{c (1 - c')}$, \cref{alg:clustering_decomp} decomposes $\cN$ in time $O(g n)$ into $g$ groups $\cG = \{B_{j, i'}\}_{j \in [g], i' \in [k]}$ of $k$ blocks each, with separation $h_{\sep}$ where $|B_{j, i'}| = \lceil\frac{c n}{w k} - \frac{n}{k'}\rceil \geq 1$, $g \leq \frac{w}{c}$, $\load_{\cG} \leq \load_{\cN}$, and such that at most $(2c + \frac{1}{c' k'} + \frac{\load_{\cG} w k}{c k'})n$ vertices are not within any block of $\cG$.
\end{theorem}

For each group, \Cref{alg:cut_strategy} then obtains pairwise perfect matchings between any two blocks (which visually correspond to a ``clique'' of blocks), and finally takes the union of these $g$-many graphs (where we allow for duplicate edges).

Our main result is the following:

\begin{theorem}[Main Result]\label{thm:cut_strategy}
    Suppose we are given a vertex set $V$, a sufficiently small $\epsilon \geq \frac{\log \log n}{\log n}$, and an algorithm which finds an $(h,s)$-hop $(\phi, \kappa)$-expander decomposition with respect to unit-demands for $h = O(1 / \epsilon^3)$, any length slack $s$, any congestion slack $\kappa$, and conductance $\phi = \frac{n^{O(\epsilon)}}{s \kappa}$.

    For any matching player, \cref{alg:cut_strategy} finds a $(t,2)$-hop $\frac{1}{2 \eta}$-expander for unit-demands with diameter at most $t$ and maximum degree at most $\Delta$.
    The algorithm terminates within $b$ iterations of the main phase and presents at most $r$ cuts to the matching player.

    We have \begin{itemize}
        \item $r = n^{O(\epsilon)}$, \hfill \small(number of cuts presented to matching player)\normalsize
        \item $b = O(1 / \epsilon)$, \hfill \small(number of iterations of main phase)\normalsize
        \item $t = O(s / \epsilon^3)$, \hfill \small(hop-length \& diameter)\normalsize
        \item $\eta = \frac{n^{O(\epsilon)}}{\phi}$, and \hfill \small(congestion)\normalsize
        \item $\Delta = n^{O(\epsilon)}$. \hfill \small(maximum degree)\normalsize
    \end{itemize}
\end{theorem}

Concretely, this means that any unit-demand can be routed along $2 t$-hop paths with congestion at most $\eta$.\footnote{Applying \cref{thm:flow character} to the main result immediately implies congestion $\tilde{O}(\eta)$, but we obtain congestion $\eta$ in our analysis.}


\begin{corollary}[Existence of a Cut-Strategy]\label{cor:cut_strategy}
The cut-strategy of \cref{thm:cut_strategy} exists for any sufficiently small $\eps \geq \frac{\log\log n}{\log n}$ with $t = O(1 / \epsilon^3)$ and congestion $\eta = n^{O(\epsilon)}$.
\end{corollary}
\begin{proof}
    By \cref{thm:expdecomp_exist}, hop-constrained expander decompositions exist for $h = O(1 / \epsilon^4)$, any $s \geq 100$, any $\phi > 0$, and congestion slack $\kappa \leq h \cdot n^{O(\frac{1}{s})} \log n$.
    If we choose $\phi = n^{O(\epsilon) - O(\frac{1}{s})}/(h s \log n) \leq \frac{n^{O(\epsilon)}}{s \kappa}$, then $\eta = s \cdot n^{O(\epsilon) + O(\frac{1}{s})}$.
    Choosing $s = 1/\eps$ yields $\eta = n^{O(\epsilon)}$.
\end{proof}

We remark that choosing a constant $\epsilon$ yields constant-hop expanders with $\eta, \Delta \approx \poly(n)$. Choosing $\epsilon = \frac{\log \log n}{\log n}$ recovers standard cut-matching games~\cite{KKOV} with $\eta, \Delta \approx \poly\log(n)$ and logarithmic hop-length.
\Cref{alg:cut_strategy} transitions smoothly between the two, and thus, allows to trade a short hop-length with small congestion and small degree. A computationally efficient cut-player for constant-hop and length-constrained expanders has has been given in \cite{LCexpanderAlg24}, as outlined in \cref{sec:results}.

\section{Proofs and Analysis}

In \cref{sec:main_phase}, we prove that the main phase completes within only ``few'' iterations. In \cref{sec:final_phase}, we prove that a large cluster in $\cN$ can be ``completed'' to a graph which can route any unit-demand along $t$-hop paths with low congestion. \Cref{thm:cut_strategy} is proven in \cref{sec:main_proof}.
We summarize the chosen parameters and their dependencies in \cref{sec:params}.

\subsection{Main Phase}\label{sec:main_phase}

To show that a cut-matching game with the given cut-strategy terminates, we show that the number of iterations of the ``main phase'' is not ``too'' large by considering a random walk on the returned matching graphs and proving that entropy strictly increases in each round.

We begin in \cref{sec:two_step_mixing_process} by describing properties of the random walk of a commodity on matching graphs. In \cref{sec:leakage}, we distinguish between ``well-behaved'' commodity flow (which we call \emph{typical}) and remaining commodity flow (which we call \emph{leaked}). We use this distinction in \cref{sec:single_entropy_increase}, to show that the entropy increases strictly with each round of the random walk.

In \cref{sec:commodities}, we introduce a joint random walk of $n$ independent commodities (one starting at each vertex) and show that before each round there are sufficiently many commodities whose leaked flow is ``small''. This and our earlier results, we then use in \cref{sec:global_entropy_increase} to show that the entropy of the joint random walk increases ``quickly'', and hence, the main phase of the cut-strategy completes within ``few'' iterations.

\subsubsection{The Two-Step Mixing Process}\label{sec:two_step_mixing_process}


We study the random walk on matching graphs within the framework of two-step mixing processes which we developed in \cref{sec:mixing_process_stable_entropy}.


We consider a single iteration of the cut-strategy where the matching player returns pairwise matchings $M = \bigcup_{j=1}^g M_j$ based on cuts from the set of groups $\cG$ with cardinality $g$.
In the following, we write $\load_M(v) = \sum_{j=1}^g \Ind{\deg_{M_j}(v) > 0} \leq \load_M$ where we use $\load_M$ synonymously with $\load_{\cG}$.
Throughout this subsection, we will refer to ``probability mass'' as ``commodity''. Within the framework of \cref{sec:mixing_process_stable_entropy}, we consider the collection of mixers of vertices $V$ given by the multiset \begin{align*}
    \cW = \bigcup_{j=1}^g \bigcup_{v \in V} \{\Gamma_{M_j}(v) \mid \text{$\Gamma_{M_j}(v) \neq \emptyset$}\}, \quad\text{and weights}\quad \vw(v) = \load_M(v) \cdot k
\end{align*} where we use $|\Gamma_{M_j}(v)| \in \{0, k\}$ to ensure that each $v \in V$ is in exactly $\vw(v)$ mixers of $\cW$.

In words, the amount of commodity at a vertex $v$ is first split equally between all matchings $M_j$ where $v$ has positive degree. If $\deg_M(v) = 0$, then commodities remain at $v$. Then, within each matching $M_j$, each vertex $v$ sends a $1/\deg_{M_j}(v)$-fraction of its commodities to all of its neighbors. This corresponds to the ``first step'' of the mixing process. Finally, each vertex $u$ mixes its commodities and sends a fraction of its commodities to every neighbor $v'$ proportionally to $\load_M(v') \cdot k$. This corresponds to the ``second step'' of the mixing process.


Let $\vp$ and $\vp'$ be the pseudo-distributions of commodity before and after the iteration of the random walk, respectively.

We denote by \begin{align*}
    \vq_0(j, v) = \begin{cases}
        \vp(v) / \load_M(v) & \text{if $\deg_{M_j}(v) > 0$} \\
        0 & \text{otherwise} \\
    \end{cases}
\end{align*} the commodity at $v$ within the matching $M_j$ before an iteration of the random walk.
Recall that $\deg_{M_j}(v) = k$ for all vertices with non-zero degree in $M_j$.
We therefore write \begin{align*}
    \vq_1(j, v, u) = \begin{cases}
        \vq_0(j, v) / k & \text{if $\{v,u\} \in M_j$} \\
        0 & \text{otherwise} \\
    \end{cases}
\end{align*} for the amount of commodity sent from $v$ to $u$ during the first step.


Observe that within each matching $M_j$, each vertex $u$ corresponds to a unique mixer $W(j,u) \in \cW$.
Recall that in the second step of the random walk a $(\gamma_{v'}/\gamma_{W(j,u)})$-fraction of commodity at $u$ after the first step is sent to $v'$ where we defined $\gamma_{v'} = 1 / \vw(v')$ and $\gamma_{W(j,u)} = \sum_{u' \in \Gamma_{M_j}(u)} \gamma_{u'}$. Therefore, \begin{align}
    \vq_2(j, v, u, v') = \begin{cases}
        \frac{\gamma_{v'}}{\gamma_{W(j,u)}} \vq_1(j, v, u) & \text{if $\{v,u\},\{u,v'\} \in M_j$} \\
        0 & \text{otherwise} \\
    \end{cases}
\end{align} is the amount of commodity that is sent from $v$ through $u$ to $v'$ within $M_j$ during an iteration of the random walk, and hence, \begin{align*}
    \vp'(v') = \begin{cases}
        \sum_{j=1}^g \sum_{\{v,u\}, \{u,v'\} \in M_j} \vq_2(j, v, u, v') & \text{if $\deg_M(v') > 0$} \\
        \vp(v') & \text{otherwise}. \\
    \end{cases}
\end{align*}

\begin{lemma}\label{lem:latent_prob}
    For any $j, v, u, v'$ such that $\vq_2(j, v, u, v') > 0$, \begin{align*}
        \vq_2(j, v, u, v') = \frac{\vp(v)}{\load_M(v) \cdot \load_M(v') \cdot k^2 \cdot \gamma_{W(j,u)}}.
    \end{align*}
\end{lemma}
\begin{proof}
    Note that the condition implies $\{v,u\},\{u,v'\} \in M_j$.
    We have \begin{align*}
        \vq_2(j, v, u, v') = \frac{\vq_0(j, v)}{\load_M(v') \cdot k^2 \cdot \gamma_{W(j,u)}} = \frac{\vp(v)}{\load_M(v) \cdot \load_M(v') \cdot k^2 \cdot \gamma_{W(j,u)}}
    \end{align*} using the definitions of $\vq_1$ and $\vq_0$.
\end{proof}

\begin{lemma}\label{lem:rw_local_dif_bound}
    For any $j, u$ such that $\deg_{M_j}(u) > 0$, $\frac{1}{\load_M} \leq \gamma_{W(j,u)} \leq 1$.
\end{lemma}
\begin{proof}
    We have \begin{align*}
        \frac{1}{\load_M} = \sum_{u' \in \Gamma_{M_j}(u)} \frac{1}{k \load_M} \leq \underbrace{\sum_{u' \in \Gamma_{M_j}(u)} \frac{1}{k \load_M(u')}}_{=\gamma_{W(j,u)}} \leq \sum_{u' \in \Gamma_{M_j}(u)} \frac{1}{k} = 1
    \end{align*} where the first inequality and final equality use $|\Gamma_{M_j}(u)| = k$, the second inequality uses $\load_M(u') \leq \load_M$, and the third inequality uses $\load_M(u) \geq 1$.
\end{proof}

\begin{lemma}\label{lem:split_value_bound}
    For any $\nu, j, v, u, v'$ such that $\vq_2(j, v, u, v') > 0$, \begin{align*}
        \frac{\vp(v)}{k^2 \load_M^2} \leq \vq_2(j, v, u, v') \leq \load_M \frac{\vp(v)}{k^2}.
    \end{align*}
\end{lemma}
\begin{proof}
    The statement follows directly using \cref{lem:latent_prob,lem:rw_local_dif_bound} and $1 \leq \load_M(v) \leq \load_M$ for all $v \in V$.
\end{proof}

Observe that $\vq_2$ is a distribution over the set $\bigcup_{v \in V} S(v)$ where \begin{align*}
    S(v) = \{(j, v, u, v') : j \in [g]; u, v' \in V; \{v,u\}, \{u,v'\} \in M_j\}.
\end{align*} Further, observe that if $\deg_M(v) > 0$, then \begin{align}
    \sum_{w \in S(v)} \vq_2(w) = \vp(v), \label{eq:merging_characterization}
\end{align} that is, the commodity at $v$ before the iteration of the random walk is ``split'' across all paths of length two starting at $v$ in any of the matchings $M_j$.

\begin{lemma}\label{lem:split_size_bound}
    For any $v \in V$ such that $\deg_M(v) > 0$, $|S(v)| \geq k^2$.
\end{lemma}
\begin{proof}
    As $\deg_M(v) > 0$, there exists at least one matching $M_j$ such that $\deg_{M_j}(v) > 0$.
    Moreover, recall that the degree of all vertices in a matching that have non-zero degree is $k$, and hence, there exist $k^2$ paths in $M_j$ starting at $v$.
\end{proof}

We define \begin{align*}
    S'(v'; \vp) = \{(j, v, u, v') : j \in [g]; u, v' \in V; \{v,u\}, \{u,v'\} \in M_j, \vp(v) > 0\}.
\end{align*} Observe that if $\deg_M(v') > 0$ it follows from \cref{lem:latent_prob} that $\vp'(v') = \sum_{w \in S'(v'; \vp)} \vq_2(w)$.

\begin{Definition}[Local commodity]
    We say that a pseudo-distribution $\vp$ is \emph{local} if for every matching $M_j$, there exists a block $B$ in the $j$-th group of $\cG$ such that $\sum_{v \in B} \vp(v) = \sum_{v \in V} \vp(v)$ (and hence, $\sum_{v \in V \setminus B} \vp(v) = 0$).

    That is, in every matching $M_j$, $\vp$ is ``locally concentrated'' in a single block.
\end{Definition}

\begin{lemma}\label{lem:merging_size_bound}
    Let $\vp$ be local.
    Then for any $v' \in V$, we have that $|S'(v'; \vp)| \leq \load_M k$.
\end{lemma}
\begin{proof}
    Consider a matching $M_j$ where $\deg_{M_j}(v') > 0$ (there are at most $\load_M$-many such matchings). We know that $\vp$ is locally concentrated in a single block $B$. The vertex $v'$ has exactly $k$ neighbors in $M_j$. If $v' \in B$, then each neighbor of $v'$ has exactly one neighbor in $B$. If $v' \not\in B$, then $v'$ itself has a neighbor in $B$, and hence, exactly $k - 1$ neighbors of $v'$ have a neighbor in $B$. We conclude that $|S'(v'; \vp)| \leq \load_M k$.
\end{proof}

\subsubsection{Leakage}\label{sec:leakage}

In this subsection, we build a characterization of commodity flow distinguishing ``typical'' and ``leaked'' commodity flow.
We continue to consider the random walk of a commodity, which we assume to ``start'' at some single vertex $v$ (to later ensure that typical commodity is local).
That is, before the first iteration of the random walk, all probability mass is located at vertex $v$.

We will use this characterization in \cref{sec:single_entropy_increase} to show that entropy increases strictly during each iteration of the random walk.

\begin{Remark}[Partial matchings]\label{rmk:removed_edges}
    We allow the matching player to remove an arbitrary $\alpha$-fraction of edges from each batch of matchings.
    We refer to this new union of matchings by $M'$.
    Thus, $M'$ is not a union of perfect matchings (as is $M$), but is obtained from $M$ by removing any $\alpha$-fraction of edges.
    We account for this in our analysis by disregarding any flow in the random walk that uses a removed edge.

    This aspect of our analysis is not used to prove the main result of this paper. However, it is useful for obtaining ``non-stop'' versions of \cref{alg:cut_strategy} where the matchings returned by the matching player may be incomplete.
\end{Remark}

\begin{Definition}[Typical commodity flow]
    We call commodity flow \emph{typical} (before an iteration of the random walk) if it satisfies the following three properties: \begin{enumerate}
    \item the commodity flow has not used any edge included in a (pure) cut $C$ of prior iterations of the random walk,
    \item the commodity is not at vertices $v$ with $\load_{\cG}(v) = 0$ (where $\load_{\cG}(v)$ denotes the number of groups in $\cG$ that include $v$), and
    \item the commodity flow does not use edges $e \in M$ which were ``removed'' as defined in \cref{rmk:removed_edges} (i.e., $e \not\in M'$).
\end{enumerate}

Commodity flow that dissatisfies any one of these properties is called \emph{leaked}.
\end{Definition}

Given the distribution of a commodity $\vp$ on $V$ (before an iteration of the random walk), we denote its leaked commodity flow by the pseudo-distribution $\hat{\vp}$, and the remaining typical commodity flow by the pseudo-distribution $\tilde{\vp} = \vp - \hat{\vp}$.

\begin{Definition}[Typical commodity]
We say that $\vp$ is \emph{$\ell$-typical} (or \emph{has leakage at most $\ell$}) if $\hat{\vp}(V) \leq \ell$ or equivalently if $\tilde{\vp}(V) \geq 1 - \ell$.
\end{Definition}

\begin{lemma}\label{lem:non_leaked_is_local}
    If blocks within a group of $\cG$ have separation at least $2 b$, then for any distribution $\vp$, $\tilde{\vp}$ is local.
\end{lemma}
\begin{proof}
    As $\tilde{\vp}$ consists only of typical commodity flow, we have that the commodity has not previously used edges that were removed by some (pure) cut $C$.
    In particular, commodity may therefore only have moved over paths of length at most $2 b$ using that the total number of iterations of the random walk is at most $b$ and in each iteration commodity may move along a path of length at most $2$ (due to the two steps of an iteration of the random walk).

    Finally, using that blocks within a group of $\cG$ have separation $h_{\sep} \geq 2 b$ and that any commodity started at a single vertex $v$, we conclude that $\tilde{\vp}$ must be concentrated at vertices within a ball of radius $2 b$ around $v$, and hence, $\tilde{\vp}$ is local.
\end{proof}

\subsubsection{Entropy Increase of a Commodity}\label{sec:single_entropy_increase}

We will now study the entropy-increase of a commodity during a single iteration of the random walk.

We denote by $\tilde{\vp}'$ and $\hat{\vp}'$  the pseudo-distributions of typical and leaked commodity flow of $\vp$ after an iteration of the random walk, respectively.\footnote{This is generally different from the typical and leaked commodity flow of $\vp'$.}
Observe that $\vp' = \tilde{\vp}' + \hat{\vp}'$.
Moreover, we denote by $\tilde{\vq}_2$ the corresponding distribution of typical commodity flow after the second step of the random walk.

We prove the entropy increase by first showing that the entropy of $\tilde{\vq}_2$ is significantly larger than the entropy of $\tilde{\vp}$ using our notion of ``splitting'' a pseudo-distribution which we developed in \cref{sec:splitting_and_merging}.
Then, we show that ``merging'' $\tilde{\vq}_2$ to obtain $\tilde{\vp}'$ reduces the entropy by a smaller amount than the original entropy increase.


\begin{lemma}[Splitting]\label{lem:entropy_inc}
If $\vp$ is an $\ell$-typical distribution, $H(\tilde{\vq}_2) \geq H(\tilde{\vp}) + (1 - \ell) \log(\frac{k^2}{\load_M})$.
\end{lemma}
\begin{proof}
    By \cref{lem:split_size_bound}, $\tilde{\vq}_2$ is obtained from $\tilde{\vp}$ by splitting each entry into at least $k^2$ new entries.
    By \cref{lem:split_value_bound}, each entry of $\tilde{\vq}_2$ obtained from the entry $\tilde{\vp}(v)$ has value at most $\load_M \frac{\tilde{\vp}(v)}{k^2}$.
    Also note that $\tilde{\vp}$ has size at least $1 - \ell$ as $\vp$ is $\ell$-typical.
    Using \cref{lem:splitting} completes the proof.
\end{proof}

\begin{lemma}[Merging]\label{lem:entropy_dec}
    If blocks within a group of $\cG$ have separation at least $2 b$, then given any distribution $\vp$, $H(\vp') \geq H(\tilde{\vq}_2) + H(\hat{\vp}') - \log(\load_M k + 1)$.
\end{lemma}
\begin{proof}
    Consider the vector $\vr$, which is obtained by concatenating $\tilde{\vq}_2$ and $\hat{\vp}'$. Observe that $\vr$ is a distribution, i.e., sums to one. Using that $\vp' = \tilde{\vp}' + \hat{\vp}'$, and the characterization of $\tilde{\vp}'$ in terms of $\tilde{\vq}_2$ (cf. \cref{eq:merging_characterization}), we have that $\vp'$ is obtained from $\vr$ by merging disjunct sets of entries of $\vr$.

    Entry $\vp'(v')$ is obtained by merging at most $|S'(v'; \tilde{\vp})|$ entries from $\tilde{\vq}_2$ and $\hat{\vp}'(v')$.
    Using that $\tilde{\vp}$ is local (cf. \cref{lem:non_leaked_is_local}), it follows from \cref{lem:merging_size_bound} that $|S'(v'; \tilde{\vp})| \leq \load_M k$.
    Thus, $\vp'(v')$ is obtained by merging at most $\load_M k + 1$ entries of $\vr$.

    It follows from \cref{lem:merging} that $H(\vp') \geq H(\vr) - \log(\load_M k + 1)$.
    Finally, observe that $H(\vr) = H(\tilde{\vq}_2) + H(\hat{\vp}')$.
\end{proof}

Using the above two lemmas and that the leaked commodity flow does not decrease in entropy due to the stable entropy of the two-step mixing process (cf. \cref{sec:mixing_process_stable_entropy}), we obtain the following iteration-wise entropy increase.

\begin{lemma}\label{lem:entropy_typical}
    If blocks within a group of $\cG$ have separation at least $2 b$ and $\vp$ is an $\ell$-typical distribution, \begin{align*}
        H(\vp') \geq H(\vp) + (1 - \ell) \log\left(\frac{k^2}{\load_M}\right) - \log(\load_M k + 1).
    \end{align*}
\end{lemma}
\begin{proof}
    We have \begin{align*}
        H(\vp') &\geq H(\tilde{\vq}_2) + H(\hat{\vp}') - \log(\load_M k + 1) \\
        &\geq H(\tilde{\vp}) + H(\hat{\vp}) + (1 - \ell) \log\left(\frac{k^2}{\load_M}\right) - \log(\load_M k + 1) \\
        &\geq H(\vp) + (1 - \ell) \log\left(\frac{k^2}{\load_M}\right) - \log(\load_M k + 1)
    \end{align*} where the first inequality follows from \cref{lem:entropy_dec}, the second inequality follows from \cref{lem:entropy_inc,thm:stable_entropy_mixing_process}, and the third inequality follows from \cref{lem:entropy_splitting} and $\vp = \tilde{\vp} + \hat{\vp}$.
\end{proof}

\subsubsection{Multiple Commodities}\label{sec:commodities}

In the following, we will consider $n$ commodities (one initially located at each vertex), which we denote by $\vp_\nu$ where $\nu \in [n]$.
We say that a commodity $\nu$ is $\ell$-typical if $\vp_\nu$ is $\ell$-typical.
We encode the joint distribution of commodities by $\mP = [\vp_1, \dots, \vp_n]$.
Note that before the first iteration of the random walk, $\mP$ is doubly stochastic by assumption.

\begin{lemma}\label{lem:doubly_stochastic}
    Given that $\mP$ is doubly stochastic, $\mP' = [\vp'_1, \dots, \vp'_n]$ is doubly stochastic.

    That is, \begin{align*}
        \sum_{v \in V} \vp'_\nu(v) = 1 \quad\forall \nu \in [n] \qquad\text{and}\qquad \sum_{\nu \in [n]} \vp'_\nu(v) = 1 \quad\forall v \in V.
    \end{align*}
\end{lemma}
\begin{proof}
    The first property follows from \cref{lem:preserving_size}.
    For the second property, observe that the two-phase random walk of \cref{sec:mixing_process_stable_entropy} is defined such that the total commodity returned from a mixer $W$ to some $v \in W$ during the second phase is equal to the total commodity sent from $v$ to $W$ during the first phase.
\end{proof}

The main result of this subsection is that before any iteration of the random walk the leakage is small for sufficiently many commodities.


\begin{lemma}\label{lem:ed_leakage}
    If $\phi \leq \frac{k}{(4 \cdot 6 \cdot 9) h s \kappa}$, then for at most a $\frac{1}{6}$-fraction of all commodities is the leakage due to edges affected by the expander decomposition greater than $\frac{1}{9}$.
\end{lemma}
\begin{proof}
    Recall that $C$ is a (pure) cut such that $|C| \leq h s \kappa \phi n$, and hence, $C$ removes at most $h s \kappa \phi n$ many edges.

    Observe that each removed edge was used during exactly one iteration of the random walk (namely, the iteration during which the edge was added). During the first phase of the random walk, the edge is used by $\frac{2}{k}$ of total commodity as each endpoint of the edge has $1$ total commodity by \cref{lem:doubly_stochastic} and the commodity is split into exactly $k$ parts. During the second phase of the random walk, the edge is again used by $\frac{2}{k}$ total commodity as its endpoints send total commodity equal to the amount sent during the first phase.

    Thus, the average leakage per commodity is at most $\frac{4 h s \kappa \phi n}{k n} \leq \frac{1}{6 \cdot 9}$. By Markov's inequality, the leakage is larger than $\frac{1}{9}$ for at most a $\frac{1}{6}$-fraction of commodities.
\end{proof}

\begin{lemma}\label{lem:rv_leakage}
    If for at most a $\beta$-fraction of vertices $v$, $\load_{\cG}(v) = 0$, and $\beta \leq \frac{1}{6 \cdot 9}$, then for at most a $\frac{1}{6}$-fraction of commodities is the leakage due to vertices $v$ with $\load_{\cG}(v) = 0$ greater than $\frac{1}{9}$.
\end{lemma}
\begin{proof}
    By \cref{lem:doubly_stochastic}, the total amount of commodity at each vertex is $1$. Hence, the average removed commodity is at most $\beta \leq \frac{1}{6 \cdot 9}$. By Markov's inequality, the leakage is larger than $\frac{1}{9}$ for at most a $\frac{1}{6}$-fraction of commodities.
\end{proof}

\begin{lemma}\label{lem:re_leakage}
    If at most a $\gamma$-fraction of edges are removed from $M$, and $\gamma \leq \frac{1}{(2 \cdot 6 \cdot 9) \load_M}$, then for at most a $\frac{1}{6}$-fraction of commodities is the leakage due to edges removed from $M$ greater than $\frac{1}{9}$.
\end{lemma}
\begin{proof}
    Analogously to the proof of \cref{lem:ed_leakage}, each removed edge is used by at most $\frac{4}{k}$ total commodity. Thus, the average leakage per commodity is at most $\frac{4}{k n} \gamma |E(M)| \leq 2 \gamma \load_M \leq \frac{1}{6 \cdot 9}$ where the first inequality uses $|E(M)| \leq \frac{n \load_M k}{2}$ which is due to the handshaking lemma. By Markov's inequality, the leakage is larger than $\frac{1}{9}$ for at most a $\frac{1}{6}$-fraction of commodities.
\end{proof}

\begin{lemma}\label{lem:typical_commodities}
    Under the conditions of \cref{lem:ed_leakage,lem:rv_leakage,lem:re_leakage}, at least a $\frac{1}{2}$-fraction of commodities is $\frac{1}{3}$-typical before any iteration of the random walk.
\end{lemma}
\begin{proof}
    The statement follows from \cref{lem:ed_leakage,lem:rv_leakage,lem:re_leakage} using a simple union bound.
\end{proof}

\subsubsection{Global Entropy Increase}\label{sec:global_entropy_increase}

We are now ready to prove the entropy increase of the joint random walk of all commodities.

\begin{theorem}\label{thm:entropy_increase}
    If blocks within a group of $\cG$ have separation at least $2 b$ and at least $\alpha \cdot n$ commodities are $\ell$-typical, then \begin{align*}
        H(\mP') \geq H(\mP) + \alpha n \left[(1 - \ell) \log\left(\frac{k^2}{\load_M}\right) - \log(\load_M k + 1)\right].
    \end{align*}
\end{theorem}
\begin{proof}
    As all commodities are independent, \begin{align*}
        H(\mP') &= \sum_{\text{$\ell$-typical $\nu$}} H(\vp'_\nu) + \sum_{\text{non-$\ell$-typical $\nu$}} H(\vp'_\nu) \\
        &\geq \sum_{\text{$\ell$-typical $\nu$}} \left[H(\vp_\nu) + (1 - \ell) \log\left(\frac{k^2}{\load_M}\right) - \log(\load_M k + 1)\right] + \sum_{\text{non-$\ell$-typical $\nu$}} H(\vp_\nu) \\
        &= H(\mP) + \alpha n \left[(1 - \ell) \log\left(\frac{k^2}{\load_M}\right) - \log(\load_M k + 1)\right]
    \end{align*} where the inequality uses \cref{lem:entropy_typical,thm:stable_entropy_mixing_process}.
\end{proof}

\begin{corollary}\label{thm:iteration_bound}
    If given any $\epsilon > 0$ such that \begin{enumerate}
        \item blocks within a group of $\cG$ have separation at least $2 b$,
        \item at least a $\frac{1}{2}$-fraction of commodities is $\frac{1}{3}$-typical before each iteration,
        \item $k \geq 21 n^\epsilon$, and
        \item $\load_M \leq n^{\frac{\epsilon}{6}}$,
    \end{enumerate} we have that $b \leq \frac{36}{\epsilon}$ where $b$ denotes the number of iterations of the ``main phase''.
\end{corollary}
\begin{proof}
    Using \cref{thm:entropy_increase} and that $H(\mP) \leq n \log n$ during every iteration, we have \begin{align*}
        b \leq \frac{\log n}{\alpha (1 - \ell) \log\left(\frac{k^2}{\load_M}\right) - \alpha\log(\load_M k + 1)}.
    \end{align*} Thus, using $\alpha \geq \frac{1}{2}$ and $\ell \leq \frac{1}{3}$, \begin{align*}
        \frac{1}{b} &\geq \frac{1}{\log n} \left[\frac{1}{3}\log\left(\frac{k^2}{\load_M}\right) - \frac{1}{2}(\load_M k + 1)\right] \\
        &\geq \frac{1}{\log n} \left[\frac{1}{6}\log k - \frac{5}{6} \log \load_M - \frac{1}{2}\right] \\
        &\geq \frac{1}{\log n} \left[\left(\frac{1}{6} - \frac{5}{36}\right) \epsilon \log n + \frac{\log 21}{6} - \frac{1}{2}\right] \\
        &\geq \frac{\epsilon}{36}
    \end{align*} where the second inequality uses $\load_M k \geq 1$ (as $n \geq 1$) and $\log(1+x) \leq 1 + \log x$ for $x \geq 1$, the third inequality uses the assumptions on $k$ and $\load_M$, and the final inequality uses $\frac{\log 21}{6} - \frac{1}{2} \geq 0$.
    It follows that $b \leq \frac{36}{\epsilon}$.
\end{proof}

\subsection{Final Phase}\label{sec:final_phase}

This subsection summarizes results for the final phase of the cut-strategy.

\begin{lemma}\label{lem:if_case:1}
    If $h \geq h_{\diam}$, any unit-demand on $S_{j,j'}$ can be routed in $G$ along $(s \cdot h)$-hop paths with congestion $O(\frac{\log n}{\phi})$. 
\end{lemma}
\begin{proof}
    Observe that $G - C$ is a $(h,s)$-hop $\phi$-expander for unit-demands and that $\diam_G(S_{j,j'}) \leq h_{\diam} \leq h$ as $S_{j,j'}$ is a cluster in a well-separated clustering of $G$ with diameter $h_{\diam}$.
    Hence, by \cref{thm:flow character}, any unit-demand on $S_{j,j'}$ can be routed in $G - C$ along $(s \cdot h)$-hop paths with congestion $O(\frac{\log n}{\phi})$. As decreasing the length of edges can make routing only easier, we have that such demands can also be routed in $G$ along $(s \cdot h)$-hop paths with congestion $O(\frac{\log n}{\phi})$.
\end{proof}

\begin{lemma}\label{lem:final_router}
    If $h \geq h_{\diam}$ and $t \geq h s + 2$, the graph $G_{\text{final}}$ returned by \cref{alg:cut_strategy} has diameter at most $t$ and routes any unit-demand along $t$-hop paths with congestion $O(k'' \frac{\log n}{\phi})$.
\end{lemma}
\begin{proof}
    We denote by $G$ the graph after the main phase of \cref{alg:cut_strategy}.

    First, observe that the diameter of $G_{\text{final}}$ is at most $h_{\diam} + 2$ where $h_{\diam}$ is the maximum diameter of clusters in the well-separated clustering $\cN$. Using $h_{\diam} \leq h$ and $h s \leq t - 2$, we conclude that the diameter of $G_{\text{final}}$ is at most $t$.

    Let $D$ be any unit-demand.
    It follows from \cref{lem:if_case:1} that if $u, v \in S_{j,j'}$, $D(u,v)$ can be routed in $t-2$ steps in $G$ as $s \cdot h \leq t - 2$.
    If $u, v \not\in S_{j,j'}$, then $D(u,v)$ can be routed in $t$ steps along the path \begin{align*}
        u \xrightarrow[M_{\text{finish}}]{} u' \xrightarrow[{G}]{} \cdots \xrightarrow[{G}]{} v' \xrightarrow[M_{\text{finish}}]{} v
    \end{align*} where $G_{\text{final}} = G \cup M_{\text{finish}}$.
    Similarly, if only one of $u$ and $v$ is in $S_{j,j'}$, then $D(u,v)$ can be routed in $t-1$ steps.

    We have that the degree of vertices $S_{j,j'}$ in $G$ increases by at most (a factor) $k''$, and hence, for any edge $e \in E_G$, $\congest_F(e) =O(k'' \frac{\log n}{\phi})$ for some flow $F$ routing $D$ in $G_{\text{final}}$.

    Finally, for any newly added edge $e \in M_{\text{finish}}$ with endpoint $u \not\in S_{j,j'}$, we have that $e$ routes at most $1$ unit of flow (as $D$ is a unit-demand), and hence, $\congest_F(e) \leq 1$ for any flow $F$ routing $D$ in $G_{\text{final}}$.
\end{proof}

\begin{lemma}\label{lem:final_matching_bound}
    $k'' \leq k'$.
\end{lemma}
\begin{proof}
    We have, \begin{align*}
        (k'' - 1) \frac{n}{k'} \leq |T_1 \cup \dots \cup T_{k''-1}| \leq |V \setminus S_{j,j'}| \leq n - \frac{n}{k'} = \frac{n(k' - 1)}{k'}.
    \end{align*} The first inequality follows from the construction of $T_1, \dots, T_{k''-1}$, the second inequality holds because $T_1 \cup \dots \cup T_{k''-1} \subseteq V \setminus S_{j,j'}$, and the third inequality is due to $S_{j,j'} \geq \frac{n}{k'}$. Solving the inequality for $k''$, completes the proof.
\end{proof}

\begin{lemma}\label{lem:max_degree}
    The maximum degree of any vertex in the union of matching graphs $G_{\text{final}}$ obtained by \cref{alg:cut_strategy} is $\Delta \leq b \cdot \load_{\cG} \cdot (k - 1) + k''$.
\end{lemma}
\begin{proof}
    We denote by $M_i$ the matching graph and by $\cG_i$ the collection of groups obtained during the $i$-th iteration of the main phase. Fix any $v \in V$. We have that $\deg_{M_i}(v) \leq \load_{M_i} (k - 1)$ for all $i$, as $v$ appears in at most $\load_{M_i} = \load_{\cG_i}$ matchings, and within each such matching $v$ is matched with at most $k - 1$ vertices from other blocks.
    As $G_{\text{final}} = M_1 \cup \dots \cup M_b \cup M_{\text{finish}}$ where $M_{\text{finish}}$ is the union of $k''$ matchings obtained in the final phase, $\deg_{G_{\text{final}}}(v) \leq \sum_{i=1}^b \deg_{M_i}(v) + k'' \leq b \cdot \load_{\cG} \cdot (k - 1) + k''$ where we choose $\load_{\cG} \geq \load_{\cG_i}$ for all $i$.
\end{proof}

\subsection{Proof of \cref{thm:cut_strategy}}\label{sec:main_proof}

\begin{proof}
    Fix any $\epsilon \geq \frac{1}{\log n}$.
    We divide the proof into separate parts for each of the statements of the main result.

    \paragraph{Obtaining an Expander-Decomposition} We assume that we are given a method to find expander decompositions with hop length $h = h_{\diam}$ (where $h_{\diam}$ is the diameter of the well-separated clustering), length slack $s$, congestion slack $\kappa$, and conductance $\phi$.

    \paragraph{Finding a Well-separated Clustering} By \cref{thm:cover-absolute-separation-existential,thm:cover-absolute-separation-alg}, we can find a well-separated clustering $\cN$ with \begin{align*}
        h_{\sep} = 2 \cdot \frac{36}{\epsilon}, \quad
        h_{\diam} = \frac{h_{\sep}}{\epsilon \epsilon'}, \quad
        \load_{\cN} = n^{O(\epsilon')} \log n, \quad\text{and}\quad
        w = \frac{n^{O(\epsilon)}}{\epsilon} \log n.
    \end{align*}
    We choose $\epsilon' = \Theta(\epsilon)$ such that $\load_{\cN} \leq n^{\frac{\epsilon}{6}}$, which also implies $h = h_{\diam} = O(1 / \epsilon^3)$.

    \paragraph{Number of Iterations of the Main Phase} We choose $k = 21 n^{\epsilon}$ and assume \begin{align*}
        \phi \leq \frac{k}{(4 \cdot 6 \cdot 9) h s \kappa} = O\left(\frac{n^{\epsilon} \epsilon^3}{s \kappa}\right).
    \end{align*}
    Note that this choice of $\phi$ satisfies the condition of \cref{lem:ed_leakage}.
    We choose \begin{align*}
        c = \frac{1}{2 \cdot 3 \cdot 6 \cdot 9}, \qquad c' = \frac{1}{2}, \qquad k' = (2 \cdot 3 \cdot 6 \cdot 9)^2 \load_{\cN} w k = \frac{n^{O(\epsilon)} (\log n)^2}{\epsilon}.
    \end{align*}
    By \cref{thm:clustering_decomp}, noting that $k' \geq \frac{w k}{c (1 - c')}$, at most a \begin{align*}
        2 c + \frac{1}{c' k'} + \frac{\load_{\cN} w k}{c k'} \leq \frac{1}{6 \cdot 9}
    \end{align*} fraction of vertices $v \in V$ are such that $\load_{\cG}(v) = 0$, and hence, the condition of \cref{lem:rv_leakage} is satisfied.
    The condition of \cref{lem:re_leakage} is satisfied trivially as we do not remove any edges from $M$. It thus follows from \cref{lem:typical_commodities} that at least a $\frac{1}{2}$-fraction of commodities is $\frac{1}{3}$-typical.
    Using $h_{\sep} \geq 2 b$ and \cref{thm:iteration_bound}, we obtain $b \leq \frac{36}{\epsilon} = O(\frac{1}{\epsilon})$.

    \paragraph{Number of Rounds} We have \begin{align*}
        r = b \cdot g \cdot {k \choose 2} + k'' \leq b \cdot \frac{w}{c} \cdot k^2 + k' = \frac{n^{O(\epsilon)} (\log n)^{2}}{\epsilon^2}
    \end{align*} where the inequality follows from $g \leq \frac{w}{c}$ (cf. \cref{thm:clustering_decomp}) and $k'' \leq k'$ (cf. \cref{lem:final_matching_bound}).

    \paragraph{Short Hop-Length and Small-Congestion} Let $t = h s + 2$.
    Using $h \geq h_{\diam}$, by \cref{lem:final_router}, $G_{\text{final}}$ has diameter at most $t$ and routes any unit-demand along $t$-hop paths with congestion $\eta$ where \begin{align*}
        \eta = O\left(\frac{k'' \log n}{\phi}\right) = \frac{n^{O(\epsilon)} (\log n)^2}{\epsilon^2 \phi}
    \end{align*} where the last bound uses $k'' \leq k'$ (cf. \cref{lem:final_matching_bound}). It follows from \cref{thm:flow character} that $G_{\text{final}}$ is a $(t,2)$-hop $\frac{1}{2 \eta}$-expander for unit-demands.

    \paragraph{Maximum Degree} By \cref{lem:max_degree}, \begin{align*}
        \deg_{G_{\text{final}}} \preccurlyeq \Delta \leq b \cdot \load_{\cG} \cdot (k-1) + k'' \leq b \load_{\cN} k + k' = \frac{n^{O(\epsilon)} (\log n)^2}{\epsilon}
    \end{align*} where the inequality uses $\load_{\cG} \leq \load_{\cN}$ (cf. \cref{thm:clustering_decomp}) and $k'' \leq k'$ (cf. \cref{lem:final_matching_bound}).

    \vspace{10pt}Finally, we note that for $\epsilon \geq \frac{\log \log n}{\log n}$, the poly-logarithmic factors in $n$ and polynomial factors in $\epsilon$ are subsumed by $n^{O(\epsilon)}$.
\end{proof}






\newcommand{\etalchar}[1]{$^{#1}$}

\appendix
\section{The KRV Reduction: Using an Efficient Cut-Strategies to Find Approximately Sparsest Cuts or Certify Expansion by Embedding an Expander}\label{sec:KRVreduction}

The KRV~\cite{KRV} reduction connecting cut-matching games to algorithms for approximate sparsest cuts is elegant and conceptually quite simple. Suppose we are interested in finding an approximately $\phi$-sparse cut in a graph $G=(V,E)$. For simplicity, assume $G$ is a constant degree graph with unit capacities and an even number of vertices $n = |V|$, and that $\frac{1}{\phi}$ is an integer.

To find such a sparse cut, we start a cut-matching game on an empty graph $G'$ with node set $V'=V$. In each round the (efficient) cut-player computes and plays a cut $S \subset V'$. Based on this bisection $(S,\bar{S})$ we create a max-flow instance in $G$ as follows: We first add a super-source node $s$ to $G$ and connect $s$ to any node in $S$ with an edge that has unit capacity. We similarly add a super-sink node $t$ connected to any node in $\bar{S}$ with a unit capacity edge. Lastly, we increase the capacity of all other edges in $G$ to $\frac{1}{\phi}$. We denote this graph with integer capacities based on the bisection $(S,\bar{S})$ and $\phi$ with $G_{S,\phi}$. We then compute a maximum integral single-commodity flow between $s$ and $t$ in $G_{S,\phi}$. Note that this flow can be of size at most $n/2$ since both $s$ and $t$ are connected to exactly $n/2$ nodes via unit capacity edges. If the size of the maximum flow between $s$ and $t$ in $G_{S,\phi}$ is $n/2$ and integral then any integral flow-path decomposition results in a set of unit-flow paths ``matching'' each node in $S$ to a unique node in $\bar{S}$. These paths furthermore have congestion $\frac{1}{\phi}$, i.e., any edge in $G$ is used by at most $\frac{1}{\phi}$ paths. In summary, any integral max-flow of size $n/2$ directly maps to a perfect matching $M$ between $S$ and $\bar{S}$ and this matching can be embedded with congestion $\frac{1}{\phi}$ in $G$. We use this perfect matching $M$ as a response of the matching-player in the cut-matching game on $G'$ and continue with the next round of the game in which the (efficient) cut-strategy produces another cut that translates to another $(s,t)$-flow problem in $G$. We continue in this way until either the cut-matching game terminates after $r$ rounds with a win for the cut-player or until one of the maximum flows is less than $n/2$. In the latter case, there must be a bottleneck between the set $S_i$ and $\bar{S_i}$ in $G$. This bottleneck must be a $\phi$-sparse cut in $G$ that can furthermore be extracted efficiently.

On the other hand, if the cut-matching game terminates, then $G'$ is guaranteed to be a sufficiently strong expander and every (matching) edge $\{u,v\}$ in $G'$ maps to a path between $u$ and $v$ in $G$ with the overall congestion of all such paths being at most $r \cdot \frac{1}{\phi}$. Overall, if all rounds succeed the above algorithm produces a $\frac{r}{\phi}$-congestion embedding of $G'$ into $G$.

In summary, coupling the matching-player in $G'$ to the flow-path-matchings computed in $G$ as described above makes it possible to translate any efficient cut-strategy into a sequence of at most $r$ bisections and related single-commodity flow problems in $G$. Either one of these bisections can be used to identify a $\phi$-sparse cut in $G$ or the flows produce an $\tO(\frac{1}{\phi})$-congestion embedding of an expander $G'$ into $G$, certifying that $G$ must be an expander with conductance close to $\phi$ itself.

\section{Clustering Decomposition}\label{sec:clustering_decomp}

This section provides a derivation of \cref{thm:clustering_decomp}.

Given a graph $G$, a \emph{group} $\cB$ is a collection of mutually disjoint vertex sets $B_1, \dots, B_{|\cB|}$, called blocks. A group has absolute separation $h_{\sep}$ if the distance between any two blocks is at least $h_{\sep}$.

A grouping $\cG$ with width $g$ and absolute separation $h_{\sep}$ is a collection of $w$ many groups $\cB_1, \dots, \cB_g$, each with absolute separation $h_{\sep}$ and such that every node appears in at least one group. We say that $\cG$ is a $\psi$-grouping if at most $(\psi \cdot |V|)$-many nodes $v$ do not appear in any group.

Lastly, we denote with $\load_{\cG}$ the maximum number of blocks (or groups) any node $v$ is contained in.







 \begin{algorithm}
 \caption{Clustering Decomposition}\label{alg:clustering_decomp}
 \begin{algorithmic}
     \Require $c \in (0,1), k \in \mathbb{N}$, well-separated clustering $\cN = \{S_{j,j'}\}_{j \in [w], j'}$ with absolute separation $h_{\sep}$ such that $|S_{j,j'}| \leq \frac{n}{k'}$
     \State $\cN' = \{S_{j,j'}\}_{j \in [w'], j'} \gets$ remove clusterings with less than $c \frac{n}{w}$ vertices from $\cN$
     \State $\{S'_{j,j'}\}_{j \in [g], j'} \gets$ split clusterings of $\cN'$ greedily (while keeping clusters intact) into clusterings of size at least $c \frac{n}{w}$
     \For{$j=1$ to $g$}
         \State Choose $j'_0 = 0, j'_1, \dots, j'_{k}$ and define $B'_{j,i'} = \bigcup_{i''=j'_{i'-1}+1}^{j'_{i'}} S'_{j,i''}$ such that $\frac{c n}{w k} - \frac{n}{k'} \leq |B'_{j,i'}| \leq \frac{c n}{w k}$.
     \EndFor
     \State \Return $\cG = \{\mathrm{take}_{\lceil\frac{c n}{w k} -\frac{n}{k'}\rceil}(B'_{j,i'})\}_{j \in [g], i' \in [k]}$
 \end{algorithmic}
 \end{algorithm}

The algorithm to decompose a well-separated clustering into groups where each group contains the same number of blocks and each block is of the same size, is described by \cref{alg:clustering_decomp}.
We use $\mathrm{take}_l(B)$ to denote the subset consisting of the first $l$ elements of $B$. In particular, in \cref{alg:clustering_decomp}, $\mathrm{take}_{\lceil\frac{c n}{w k} - \frac{n}{k'}\rceil}(B'_{j,i'})$ removes vertices originating from at most one cluster $S_{j,i''}$.

\begin{lemma}\label{lem:clustering_decomp:1}
    Removing clusterings with less than $c \frac{n}{w}$ vertices from the well-separated clustering $\cN$, retains a $c$-grouping with absolute separation $h_{\sep}$.
\end{lemma}
\begin{proof}
    As $w$ is the number of clusterings in $\cN$, we remove at most $w \cdot c \frac{n}{w} = c n$ vertices.
\end{proof}


\begin{lemma}\label{lem:clustering_decomp:3}
    $\{S'_{j,j'}\}_{j \in [g], j'}$ is a $2c$-grouping with width $g \leq \frac{w}{c}$ and absolute separation $h_{\sep}$, and where the size of a group is at least $c \frac{n}{w}$ and at most $c \frac{n}{w} + \frac{n}{k'}$.
\end{lemma}
\begin{proof}
    First, observe that splitting clusterings into clusters of size at least $c \frac{n}{w}$, we drop at most $c \frac{n}{w}$ per clustering, hence at most $c n$ in total.
    Using \cref{lem:clustering_decomp:1} and that splitting clusterings (while keeping clusters intact) preserves the absolute separation property, it follows that $\{S'_{j,j'}\}_{j \in [g], j'}$ is a $2c$-grouping with absolute separation $h_{\sep}$.

    Using that the size of a cluster is at most $\frac{n}{k'}$, a group contains at most $c\frac{n}{w} + \frac{n}{k'}$ vertices.

    Finally, a clustering contains at most $n$ vertices and the size of a cluster is at least $c \frac{n}{w}$. Hence, $g \leq \frac{w}{c}$.
\end{proof}

\begin{lemma}\label{lem:clustering_decomp:4}
    Let $c' \in (0,1)$ such that $k' \geq \frac{w k}{c (1 - c')}$.
    For any group $j \in [g]$ with size in $[\frac{c n}{w} : \frac{c n}{w} + \frac{n}{k'}]$ and whose blocks have size at most $\frac{n}{k'}$, we can find $j'_1, \dots, j'_k$ such that $1 \leq \frac{c n}{w k} - \frac{n}{k'} \leq |B'_{j,i'}| \leq \frac{c n}{w k}$ and at most $\frac{1}{c' k'} n$ vertices are dropped using time $O(n)$.
\end{lemma}
\begin{proof}
    First, observe that the blocks $B'_{j,i'}$ are disjoint by construction.
    Given the clustering $\{S_{j,j'}\}_{j'}$, define $l_{i'} := |B'_{j,i'}| = \sum_{i''=j'_{i'-1}+1}^{j'_{i'}} |S_{j,i''}|$ where $j'_0 = 0$. We then select $j'_1, \dots, j'_{k}$ greedily such that $l_{i'} \geq \frac{c n}{w k} - \frac{n}{k'}$. As $|S_{j,j'}| \leq \frac{n}{k'}$, we have that $l_{i'} \leq \frac{c n}{w k}$.
    Observe that $\frac{c n}{w k} \geq \frac{c (1-c') n}{w k} \geq \frac{n}{k'}$, and hence, $l_{i'} \geq 1$.

    The group contains at least $\frac{c n}{w}$ vertices and each block has size at most $\frac{c n}{w k}$. Hence, there must be at least $k$ blocks.

    Finally, observe that $\frac{n}{k'} \leq \frac{c (1 - c') n}{w k}$, and hence, there are at most \begin{align*}
        \frac{\frac{c n}{w} + \frac{n}{k'}}{\frac{c n}{w k} - \frac{n}{k'}} \leq \frac{\frac{c n}{w} + \frac{n}{k'}}{\frac{c c' n}{w k}} \leq k + \frac{w k}{c c' k'}
    \end{align*} blocks where the second inequality uses $c' < 1$.
    Therefore, by dropping all blocks after the $k$-th block, we drop at most \begin{align*}
        \frac{w k}{c c' k'} \cdot \frac{c n}{w k} = \frac{n}{c' k'}
    \end{align*} vertices.
    As there are at most $n$ clusters $S_{j,j'}$, the greedy selection process completes in time $O(n)$.
\end{proof}

\begin{lemma}\label{lem:clustering_decomp:5}
    Given a $\psi$-grouping $\{B'_{j,i'}\}_{j \in [g], i' \in [k]}$ with absolute separation $h_{\sep}$ where $|B'_{j,i'}| \leq \frac{c n}{w k}$ and the size of a group is at least $\frac{c n}{w}$, \begin{align*}
        \{\mathrm{take}_{\lceil\frac{c n}{w k} - \frac{n}{k'}\rceil}(B'_{j,i'})\}_{j \in [g], i' \in [k]}
    \end{align*} is a $(\psi + \frac{\load_{\cG} w k}{c k'})$-grouping with absolute separation $h_{\sep}$.
\end{lemma}
\begin{proof}
    We remove at most $\frac{n}{k'}$ nodes from each block $B'_{j,i'}$, as the size of a block is bounded by $\frac{c n}{w k}$. Therefore, within group $j \in [g]$, we remove at most $\frac{k}{k'} n$ vertices.

    The number of removed vertices within group $j$ relative to a vertex in group $j$ is at most \begin{align*}
        \frac{\frac{k}{k'} n}{\frac{c n}{w}} = \frac{w k}{c k'} =: b
    \end{align*} as the number of vertices in a group is at least $\frac{c n}{w}$. We say that $b$ is the ``blame'' of a vertex with respect to a single group. The total number of removed vertices among all groups relative to any fixed vertex (i.e., the ``total blame'' of a vertex) is at most $\tilde{b} := \load_{\cG} b$. We removed at most $n \tilde{b} = \frac{\load_{\cG} w k}{c k'} n$ many vertices.

    Moreover, as we have noted before, removing vertices preserves the absolute separation property.
\end{proof}

The above lemmas can be combined to prove the main result of this section.

\begin{proof}[Proof of \cref{thm:clustering_decomp}]
    It follows from \cref{lem:clustering_decomp:3} that the second step of \cref{alg:clustering_decomp} yields a $2c$-grouping with absolute separation $h_{\sep}$ and width $g \leq \frac{w}{c}$ and where each clustering is of size in $[\frac{c n}{w} : \frac{c n}{w} + \frac{n}{k'}]$.
    Using \cref{lem:clustering_decomp:4}, $\{B'_{j,i'}\}_{j \in [g], i' \in [k]}$ is a $(2c + \frac{1}{c' k'})$-grouping with absolute separation $h_{\sep}$ (as merging blocks preserves the absolute separation property) and width $g$ such that $\frac{c n}{w k} - \frac{n}{k'} \leq |B'_{j,i'}| \leq \frac{c n}{w k}$.
    Finally, using \cref{lem:clustering_decomp:5}, $\{B_{j,i'}\}_{j \in [g], i' \in [k]}$ is a $(2c + \frac{1}{c' k'} + \frac{\load_{\cG} w k}{c k'})$-grouping with absolute separation $h_{\sep}$ where $|B_{j,i'}| = \lceil\frac{c n}{w k} - \frac{n}{k'}\rceil$. Note that $\load_{\cG} \leq \load_{\cN}$ as we did not add any additional vertices.

    By \cref{lem:clustering_decomp:4}, the for-loop takes at most $O(g n)$ time.
\end{proof}

\section{Overview of Parameters}\label{sec:params}

We are fixing a vertex set $V$ of size $n$ and a sufficiently small parameter $\epsilon \geq \frac{\log\log n}{\log n}$ trading hop-length and congestion.

We denote by $\epsilon_{\mathrm{C}}$ and $\epsilon'_{\mathrm{C}}$ the parameters of the well-separated clustering $\cN$ (cf. \cref{thm:cover-absolute-separation-existential}).

\begin{center}
\begin{tabular}{ |c|l|l|p{0.4\textwidth}| }
    \hline
    Ref & Parameter & Description & Constraints \\
    \hline\hline
    A & $k = 21 n^{\epsilon}$ & \# blocks in $\cG$ & \makecell[tl]{\textbullet\ $k \geq 21 n^{\epsilon}$ (\cref{thm:iteration_bound}) \\ \quad $\rightarrow$ entropy increase} \\
    B & $\load_{\cN} \leq n^{\frac{\epsilon}{6}}$ & load of $\cN$ & \makecell[tl]{\textbullet\ for some $\epsilon'_{\mathrm{C}} = \Theta(\epsilon)$ (\cref{sec:main_proof})} \\
    C & $b \leq \frac{36}{\epsilon}$ & \# iterations of ``main phase'' & \makecell[tl]{\textbullet\ from (A), (B), (D) (\cref{thm:iteration_bound})} \\
    D & $h_{\sep} = 2 \cdot \frac{36}{\epsilon}$ & separation of $\cN$ & \makecell[tl]{\textbullet\ $h_{\sep} \geq 2 b$ (\cref{lem:non_leaked_is_local}) \\ \quad $\rightarrow$ non-leaked commodity is local} \\
    E & $h_{\diam} = \frac{h_{\sep}}{\epsilon_{\mathrm{C}} \epsilon'_{\mathrm{C}}}$ & diameter of $\cN$ & \makecell[tl]{\textbullet\ due to \cref{thm:cover-absolute-separation-existential}} \\
    F & $h = h_{\diam}$ & hop-length of ED $C$ & \makecell[tl]{\textbullet\ $h \geq h_{\diam}$ (\cref{lem:if_case:1}) \\ \quad $\rightarrow$ small diameter} \\
    G & $\phi \leq O\left(\frac{n^{\epsilon} \epsilon^3}{s \kappa}\right)$ & conductance of $G - C$ & \makecell[tl]{\textbullet\ $h s \kappa \phi n \leq \frac{k n}{4 \cdot 6 \cdot 9}$ (\cref{lem:ed_leakage}) \\ \quad $\rightarrow$ small leakage due to ED} \\
    H & $t = h s + 2$ & hop-length of $G_{\mathrm{final}}$ & \makecell[tl]{\textbullet\ $t \geq h s + 2$ (\cref{lem:final_router})} \\
    I & $k' = O(\load_{\cN} w k)$ & ``final phase'' threshold & \makecell[tl]{\textbullet\ $k' \geq \Theta(w k)$ (\cref{thm:clustering_decomp}) \\ \textbullet\ small leakage due to grouping \\ \; (\cref{lem:rv_leakage})} \\
    J & $r \leq b \cdot k^2 \cdot g$ & \# cuts &  \\
    K & $\Delta \leq b \cdot \load_{\cN} \cdot (k - 1) + k'$ & maximum degree of $G_{\mathrm{final}}$ & \makecell[tl]{\textbullet\ see \cref{lem:max_degree}} \\
    L & $\eta \leq \Tilde{O}(\frac{k'}{\phi})$ & congestion in $G_{\mathrm{final}}$ & \makecell[tl]{\textbullet\ see \cref{lem:final_router}} \\
    \hline
\end{tabular}
\end{center}

\end{document}